\documentclass[11pt]{article}

\usepackage[colorlinks=false]{hyperref}
\usepackage[dvips]{graphicx}
\usepackage[usenames,dvipsnames]{color}
\usepackage{epsfig}
\usepackage{rotating}
\usepackage{amssymb}
\usepackage[boxed]{algorithm}
\usepackage[latin1]{inputenc}
\usepackage{verbatim}  
\usepackage{lscape} 

\newenvironment{proof}{{\em Proof:}}
{\hspace{\stretch{1}}%
\rule{1ex}{1ex}\\}

\begin{document}

\newcommand{\beg}{\begin{equation}}
\newcommand{\eeg}{\end{equation}}
\newcommand{\br}{\begin{array}}
\newcommand{\er}{\end{array}}
\newcommand{\bea}{\begin{equation} \begin{array}{c}}
\newcommand{\eea}{\end{array} \end{equation}}
\newcommand{\un}{\underline}
\newcommand{\cb}{\begin{center}}
\newcommand{\ce}{\end{center}}
\newcommand{\llg}{\left\langle}
\newcommand{\rrg}{\right\rangle}
\newcommand{\thh}{$^{\mbox{th}}$}
\newcommand{\al}{\alpha}
\newcommand{\bite}{\begin{itemize}}
\newcommand{\eite}{\end{itemize}}
\newcommand{\m}{\mbox}
\newcommand{\wh}{\hspace{3mm} \mbox{ where }}
\newcommand{\lcm}{\mbox{lcm}}
\newcommand{\di}{\mbox{ div }}
\newcommand{\T}{\mbox{ TRUE }}
\newcommand{\SP}{\mbox{ SPAN}}
\newcommand{\GF}{\mbox{ GF}}
\newcommand{\gf}{\mbox{ {\tiny{GF}}}}
\newcommand{\GR}{\mbox{ GR}}
\newcommand{\Di}{\mbox{ DIV}}
\newcommand{\mt}{\hspace{10mm} \mbox{ }}
\newcommand{\mf}{\hspace{5mm} \mbox{ }}
\newcommand{\mz}{\hspace{2mm} \mbox{ }}
\newcommand{\mv}{\hspace{1mm} \mbox{ }}
\newcommand{\ga}{\gamma}
\newcommand{\xl}{\begin{tiny} \begin{array}{c} < \\ \simeq \end{array} \end{tiny}}
\newcommand{\xm}{\begin{tiny} \begin{array}{c} > \\ \simeq \end{array} \end{tiny}}
\newcommand{\OR}{\mbox{ ord}}
\newcommand{\dg}{\mbox{ deg}}
\newtheorem{thm}{Theorem}
\newtheorem{pro}{Proposition}
\newtheorem{lem}{Lemma}
\newtheorem{cj}{Conjecture}
\newtheorem{cor}{Corollary}
\newtheorem{df}{Definition}
\newtheorem{imp}{Implication}
\newcommand{\mo}{\mbox{ mod }}
\newcommand{\Tr}{\mbox{ Tr}}
\newcommand{\mn}{\mbox{ {\tiny{min}}}}
\newcommand{\mx}{\mbox{ {\tiny{max}}}}
\newcommand{\erf}{\mbox{ erf}}
\newcommand{\erfc}{\mbox{ erfc}}
\newcommand{\SNR}{\mbox{ SNR}}
\newcommand{\BER}{\mbox{ BER}}
\newcommand{\hl}{\\ \hline}
\newcommand{\bd}[1]{\textbf{#1}}

\title{From Graph States to Two-Graph States}
\author{Constanza Riera\thanks{C. Riera and M. G. Parker are with the
  Selmer Centre, Inst. for Informatikk, H{\o}yteknologisenteret i Bergen,
  University of Bergen, Bergen 5020, Norway. E-mail: \texttt{riera@ii.uib.no,matthew@ii.uib.no}.
  Web: \texttt{http://www.ii.uib.no/$\sim$riera,matthew/}. C. Riera was supported in part by the Norwegian
Research Council.} and St\'ephane Jacob\thanks{St\'ephane Jacob is at L'\'ecole polytechnique, France
E-mail: \texttt{stephane.jacob@polytechnique.edu}}
and Matthew G. Parker}

\date{\today}
\maketitle

\begin{abstract}
The name `{\em graph state}' is used to describe a certain class of pure quantum state
which models a physical structure on which one can perform {\em measurement-based quantum
computing}, and which has a natural graphical description.
We present the {\em{two-graph state}}, this being a generalisation of the {\em{graph state}}
and a two-graph representation of a {\em stabilizer state}.
Mathematically, the two-graph state can be viewed as
a simultaneous generalisation of a  binary linear code and quadratic Boolean function.
It describes precisely
the coefficients of the pure quantum state vector resulting from the action of a member of the {\em{local
Clifford group}} on a graph state, and comprises a graph which encodes the
{\em{magnitude}} properties of
the state, and a graph encoding its {\em{phase}} properties. This description
facilitates a computationally
efficient spectral analysis of the graph state with respect to operations from the
local Clifford group on the state, as
all operations can be realised graphically. By focusing on the so-called
{\em{local transform group}},
which is a size 3 cyclic subgroup of the local Clifford group over one qubit, and
over $n$ qubits is of size $3^n$,
we can efficiently compute spectral properties of the graph state.
\end{abstract}

\section{Introduction}
\subsection{Codes with phase}
Consider a binary linear code, ${\cal C}$, of length $n$ and dimension $k$. We can represent
${\cal C}$ by its {\em indicator vector} in $({\mathbb{Z}}_2^2)^{\otimes n}$,
${\cal I}_m = (m(0\ldots 0),
m(0\ldots 1),\ldots,m(1\ldots 1)) = (m({\bf{x}}))$, where $m$, the {\em indicator function},
is a mapping from ${\mathbb{Z}}_2^n \rightarrow {\mathbb{Z}}_2$
such that $m({\bf{x}}) = 1$ iff
${\bf{x}} \in {\cal C}$, otherwise $m({\bf{x}}) = 0$. The indicator vector is, therefore,
the truth-table of $m$. For example,
the $n = 3$, $k = 2$ binary linear
code, with codewords ${\cal C} = \{000,011,110,101\}$, can be represented by the indicator
vector ${\cal I}_m = (1,0,0,1,0,1,1,0)$.
The indicator function is a {\em Boolean function} and respects
a non-unique factorization,
$m({\bf{x}}) = \prod_{i=0}^{n-k-1} m_i({\bf{x}})$, where the Boolean functions,
$m_i$, are affine functions, i.e. of algebraic degree $\le 1$ if ${\cal C}$ is linear,
in which case each function, $m_i$, represents
the row of a parity-check matrix that defines ${\cal C}$. For instance, for the above example,
$m({\bf{x}}) = (x_0 + x_1 + x_2 + 1)$. As another example, if
${\cal C} = \{010,101\}$, then ${\cal I}_m = (0,0,1,0,0,1,0,0)$ and $m$
can be written as $m({\bf{x}}) = (x_0 + x_1)(x_0 + x_2 + 1)
 = (x_0 + x_1)(x_1 + x_2) = (x_0 + x_2 + 1)(x_1 + x_2)$ where, in this case, ${\cal C}$ is a
{\em coset code} as it is a binary linear code additively offset by the codeword $010$.
By placing the `ones' in different positions in ${\cal I}_m$,
one can, more generally, represent any binary nonlinear code,
where $m$ is no longer the product of affine factors. We do not consider such generalisations
in this paper but we do consider another generalisation where a $\pm 1$ phase can be applied
to every entry of ${\cal I}_m$ - thus we consider codes where every codeword has an associated
phase. In order to accomodate such a generalisation we introduce the indicator
vector, $\left|\psi\right> = \frac{1}{\sqrt{w(m)}}(m(0\ldots 0)(-1)^{p(0\ldots 0)},
m(0\ldots 1)(-1)^{p(0\ldots 1)},\ldots,
m(1\ldots 1)(-1)^{p(1\ldots 1)}) = (\frac{1}{\sqrt{w(m)}}m({\bf{x}})(-1)^{p({\bf{x}})})$,
being a vector
in $({\mathbb C}^2)^{\otimes n}$, where $w(m)$ is the support weight of $m$
(i.e. the number of `ones' in the truth-table of $m$),
and $m$ and $p$ are Boolean functions from
${\mathbb{Z}}_2^n \rightarrow {\mathbb{Z}}_2$, although
we embed the ${\mathbb{Z}}_2$ output of $m$ into $\{0,1\}$
 of the complex numbers. With such
a definition, $\left|\psi\right>$ is normalised such that
$\sum_{{\bf{x}} \in {\mathbb{Z}}_2^n} |\left|\psi\right>_{{\bf{x}}}|^2 = 1$ and
the codeword ${\bf{x}}$ can be considered to be sampled from the code, ${\cal C}$, defined
by $\left|\psi\right>$, with probability $|\left|\psi\right>_{{\bf{x}}}|^2$.
In this paper we focus on the case where $m$ is a product of
affine Boolean functions and $p$ is a quadratic Boolean function. For example,
the $n = 3$, $k = 2$ binary linear
`code-with-phase' comprising codewords ${\cal C} = \{+000,-011,-110,-101\}$,
can be represented by the indicator
vector $\left|\psi\right> = \frac{1}{2}(1,0,0,-1,0,-1,-1,0) = \frac{1}{2}m(-1)^p =
\frac{1}{2}(x_0 + x_1 + x_2 + 1)(-1)^{x_0x_1 + x_2} =
\frac{1}{2}(x_0 + x_1 + x_2 + 1)(-1)^{x_1x_2 + x_1 + x_2}
 = \frac{1}{2}(x_0 + x_1 + x_2 + 1)(-1)^{x_0x_2 + x_0 + x_2}
 = \frac{1}{2}(x_0 + x_1 + x_2 + 1)(-1)^{x_0x_1 + x_0 + x_1}$.
For a given $m$ function,
there will, in general, be more than one choice of $p$ function. The choice of letters,
$m$ and $p$, is to remind the reader that $m$ assigns `magnitude' to the codewords in the
code, and $p$ assigns `phase'.
Later in this paper
we shall need to generalise to indicators of the form
$\left|\psi\right> = (\frac{1}{\sqrt{w(m)}}m({\bf{x}})i^{p({\bf{x}})})$ where $m$ is, once
again, a product of affine Boolean functions, but now $p$ is a generalised quadratic
Boolean function from ${\mathbb Z}_2^n \rightarrow {\mathbb Z}_4$ of the `special form'
$p({\bf{x}}) = (\sum_{i < j} a_{ij}x_ix_j) + (\sum_j b_jx_j) + c$, were
$a_{ij} \in \{0,2\}$, and $b_j,c \in \{0,1,2,3\}$.


\subsection{Quantum states and the local Clifford group}
The use of `bra-ket' notation, $\left| * \right>$, to denote the code-with-phase indicator
is because $\left|\psi\right>$ can be interpreted as the description for a pure quantum
state vector of $n$ qubits with the property that the $n$ qubits described by
$\left|\psi\right>$ are projected into state ${\bf{x}}$ with probability
$|\left|\psi\right>_{{\bf{x}}}|^2$ by a joint measurement of
$\left|\psi\right>$ in the so-called `computational basis' \cite{NC:QC}.
We shall show
(corollary \ref{corthm}) that, by restricting $m$ to
a product of affine functions, and $p$ to a generalised quadratic Boolean function of the
special form described previously, $\left|\psi\right>$ describes,
exactly, the class of quantum
{\em stabilizer states} for qubits \cite{Cald:Qua,Gott:Stab}.

Two pure $n$-qubit states,
$\left|\psi'\right>$ and $\left|\psi\right>$, are considered {\em locally-equivalent} if there exists
a $2^n \times 2^n$ unitary matrix, $U$, with tensor factorisation
$U = U^{(0)} \otimes U^{(1)} \otimes \ldots \otimes U^{(n-1)}$, where each $U^{(i)}$
is a $2 \times 2$
unitary matrix, such that $\left|\psi'\right> = U\left|\psi\right>$. In the context of
quantum information, local equivalence preserves the
structure of the $n$-partite quantum system, in particular the $n$-partite
{\em entanglement} of
the system \cite{NC:QC}. An important group of $2 \times 2$ unitary matrices
is the (complex) {\em local Clifford group}, ${\bf C_1}$ which
can be generated by the {\em Hadamard matrix},
$H=\frac{1}{\sqrt{2}}\begin{tiny}\left(\begin{array}{rr}1&1\\ 1&-1\end{array}\right)\end{tiny}$,
and the {\em negahadamard matrix},
$N=\frac{1}{\sqrt{2}}\begin{tiny}\left(\begin{array}{rr}1&i\\ 1&-i\end{array}\right)\end{tiny}$, where
$i = \sqrt{-1}$. The $n$-qubit local Clifford group is then given by ${\bf C_n} = {\bf C_1}^{\otimes n}$.
A {\em graph state} is of the form $\left|\psi\right> = 2^{\frac{-n}{2}}(-1)^p$, where
$p$ is a homogeneous quadratic Boolean function and, implicitly, $m = 1$.
When $m = 1$, all $\left|\psi\right>_{\bf x}$ have the same magnitude, and
we refer to such state vectors, $\left|\psi\right>$, as {\em{flat}} \cite{RP:BC1}.
The homogeneous quadratic, $p$,
maps, bijectively, to a simple graph \cite{thesis}.
It can be shown that
every stabilizer state is locally equivalent to a set of
graph states, where each such graph state is obtained
via the action of a specific unitary from ${\bf C_n}$ on the stabilizer.
In this paper we represent stabilizer states by the form
$\left|\psi\right> = 2^{\frac{-n}{2}}mi^p$, where $p$ is quadratic of the special form and
$m$ is a product of affine Boolean functions \cite{Par:QE}. This form is a generalisation of
that for a graph state. As $m$ is the indicator function for
a binary linear coset code, it can be represented by a bipartite graph with loops,
as will be made clear
later \cite{Par:QE,Dan:Pivot}. As both $m$ and $p$ can, with minor embellishments,
be represented by graphs, we refer to
$\left|\psi\right>$ of this form as a {\em two-graph state} and the two-graph state is a
bi-graphical representation of a stabilizer state.

\subsection{The Pauli group, stabilizer states, and graph states}
The single-qubit {\em Pauli group} of matrices, ${\bf P_1}$, is generated by
$X = \begin{tiny} \left ( \begin{array}{rr}
0 & 1 \\
1 & 0
\end{array} \right ) \end{tiny}$,
$Z = \begin{tiny} \left ( \begin{array}{rr}
1 & 0 \\
0 & -1
\end{array} \right ) \end{tiny}$, and $i$,
and the Pauli group for $n$ qubits is
${\bf P_n} = {\bf P_1}^{\otimes n}$. 
Formally, a stabilizer state over a system of $n$ qubits is
defined to be a joint eigenvector of a stabilizer generated by
a certain subgroup of ${\bf P_n}$ \cite{Cald:Qua,Gott:Stab,DanQECC}.
A graph state is a special case of a stabilizer state, being a joint eigenvector of
a subgroup of ${\bf P_n}$, and the graph state can be described by the edges of a simple graph with $n$ nodes
\cite{Raus:QC,Sch:QG,Hein:Graph}. The stabilizer generated by a subgroup of the Pauli group came to prominence in the mid-90's when it was used to describe a class of quantum error-correcting codes \cite{Cald:Qua,Gott:Stab}. In this context the stabilizer state describes a quantum error-correcting code of zero dimension which is robust to errors caused by a convex combination of members of the Pauli group.

It has been shown
in \cite{RP:BC1,VanD:Gr} that the graph state can always be represented
by a homogeneous
quadratic Boolean function whose structure can be bijectively mapped to the associated
graph in an obvious way.
Although the graph state has its origins in the theory of eigensystems,
its re-interpretation as a quadratic Boolean function
allows one to consider new cryptographic criteria for the function,
such as its {\em generalised bentness} \cite{RP:BC1,thesis}, or
{\em{aperiodic propagation criteria}} \cite{DanAPC}, and to justify applying such criteria to
Boolean functions of higher degree. In this paper we express the stabilizer state as
a two-graph state, this
being a simultaneous generalisation of a binary linear coset code and a
quadratic Boolean function. Such a generalisation shall allow us, in future work, to propose
and investigate new criteria for binary linear codes, and also to establish unforeseen links
between Boolean functions and coding theory.
Stabilizer states also have a natural interpretation as
GF$(4)$ additive codes \cite{Cald:Qua} and the analysis of graph states relates
naturally to recent graph-theoretic results for the associated graphs \cite{RP:Int}.

\subsection{The action of the local Clifford group}
Apart from highlighting the two-graph, magnitude-phase form of the stabilizer state,
the primary purpose of this paper is to efficiently
describe how the action of unitary matrices from the
local Clifford group, ${\bf C_n}$, modify the form of the two-graph state.
In particular, we focus
on efficiently computing spectral metrics of the form
$\sum_{U \in {\bf C_n}} \sum_{\bf{x}} |U\left|\psi\right>_{\bf{x}}|^j$ for
some integer $j$. In such cases one is only interested in the magnitudes of the elements
of $U\left|\psi\right>$, not their phases, and this simplification allows us to further
simplify as we only need to sum over all $U$ in a size $3^n$ subgroup, ${\bf T_n}$,
of the local Clifford group, as shall be explained later.
It has been shown in previous work
\cite{Bou:Tree,DanQECC,DanClass,Dan:Pivot,Glynn:Graph,RP:BC1,RP:Piv,thesis,VanD:Gr,VanD:Pivot}
how the action of matrices from the local Clifford group
on the graph state can be realised
using only local graphical operations, where linear phase offsets, generated by each matrix action,
are repeatedly eliminated by invoking local equivalence.
These two graphical
operations are called {\em edge-local complementation} (ELC)
(sometimes called {\em pivot}),
and {\em local complementation} (LC), where ELC can be decomposed into a series of LCs.
Whilst ELC acting on bipartite graphs
can be used to classify binary linear codes \cite{Dan:Pivot},
LC acting on graphs can
be used to classify additive codes over GF$(4)$ \cite{DanQECC}.
In this paper, ELC and LC are generalised so as to realise the
action of matrices from the local Clifford group
on the two-graph state, without the requirement to repeatedly eliminate linear phase offsets.

\subsection{Example}
Here is a small example that should clarify some of the ideas discussed so far:

Consider the $n = 3$-qubit graph state, $\left| \psi \right >$, which is
the joint eigenvector of the
group of commuting operators, 
$\langle {\cal{K}}_0,{\cal{K}}_1,{\cal{K}}_2 \rangle$, where ${\cal{K}}_0 = X\otimes Z\otimes I$, ${\cal{K}}_1 = Z\otimes X\otimes Z$,
${\cal{K}}_2 = I\otimes Z\otimes X$, and $I$ is the $2 \times 2$ identity matrix. Then $\left| \psi \right >$
can be represented by the simple graph, $P = ({\cal V}_P,{\cal E}_P)$, with vertices
${\cal V}_P = \{0,1,2\}$ and edges ${\cal E}_P = \{01,12\}$.
The state $\left| \psi \right >$ can be written explicitly in the
computational basis as $\frac{1}{\sqrt{8}}(\left| 000 \right > + \left| 001 \right > + \left| 010 \right > - \left| 011 \right >
 + \left| 100 \right > + \left| 101 \right > - \left| 110 \right > + \left| 111 \right >)$, which we abbreviate to
$\left| \psi \right > = \frac{1}{\sqrt{8}}(1,1,1,-1,1,1,-1,1)$, and
can alternatively be written,
using {\em{algebraic normal form}} (ANF) for the phase,
as $\left| \psi \right > = (\frac{1}{\sqrt{8}}(-1)^{p({\bf{x}})}) =
(\frac{1}{\sqrt{8}}(-1)^{x_0x_1 + x_1x_2})$, where
$p : {\mathbb Z}_2^3 \rightarrow {\mathbb Z}_2$ , and
$\left| \psi \right >_{\bf x} = \frac{1}{\sqrt{8}}(-1)^{p({\bf{x}})}$.
The quadratic monomial $x_ix_j$ is a term in $p$ iff $ij$ is an edge in $P$.
Let
$\left| \psi' \right > = (I \otimes N \otimes I)\left| \psi \right > = \frac{\omega}{\sqrt{8}}(-1)^{x_0x_1 + x_0x_2 + x_1x_2}
i^{3(x_0 + x_1 + x_2)}$, where $i = \sqrt{-1}$, and $\omega = e^{\pi i/4}$.
Then $\left| \psi' \right >$ is flat, and
the quadratic part of $\left| \psi' \right >$ represents the
graph $P'$ with edge set ${\cal{E}}_{P'} = \{01,02,12\}$ - the affine part of
$\left| \psi' \right >$
can be eliminated by subsequent action of the diagonal unitary,
$D = \omega^7
\begin{tiny}\left(\begin{array}{rr}1&0\\ 0&i\end{array}\right)
\otimes \left(\begin{array}{rr}1&0\\ 0&i\end{array}\right)
\otimes \left(\begin{array}{rr}1&0\\ 0&i\end{array}\right)\end{tiny}$, which is
in ${\bf C_n}$.
The state,
$\left| \psi' \right >$ is, by construction,
local unitary equivalent, via unitaries from ${\bf C_n}$,
to the graph state $\left| \psi \right >$, and therefore represents, to within local equivalence,
the same stabilizer state as $\left| \psi \right >$.
A graphical way of
interpreting the action of $D(I \otimes N \otimes I)$
on $\left| \psi \right >$ is to perform the action of local complementation 
on $P$ at vertex $1$ to produce graph $P'$, that is to complement all edges between
the neighbours of vertex $1$. This example shows how the action of a
unitary from the local Clifford group
maps between two locally-equivalent graph states. But, let us now consider
$\left| \psi'' \right > = (H \otimes I \otimes I)\left| \psi \right > = \frac{1}{2}(1,0,0,1,1,0,0,-1)$,
which, by construction,
is the same stabilizer state as $\left| \psi \right >$, to within local equivalence,
but not a graph state
as we cannot represent $\left| \psi'' \right >$ using only a quadratic Boolean function for its phase part.
But we can represent $\left| \psi'' \right >$ using a two-ANF representation:
$$ \left| \psi'' \right > = \frac{1}{2}m''({\bf{x}})(-1)^{p''({\bf{x}})} = \frac{1}{2}(x_0 + x_1 + 1)(-1)^{x_1x_2}, $$
where $\left| \psi'' \right >_{\bf{x}} = m''({\bf{x}})(-1)^{p''({\bf{x}})}$, $m'': {\mathbb Z}_2^2 
\rightarrow {\mathbb Z}_2$, and
$p'': {\mathbb Z}_2^2 \rightarrow {\mathbb Z}_2$.
As mentioned previously,
throughout this paper we perform a final embedding of the output of $m$,
namely ${\mathbb Z}_2$, into the complex, $\{0,1\}$,
so as to interpret the two-ANF state as a pure quantum state. To keep notation
simple, we shall not formally indicate this embedding. We refer to this two-ANF representation as
an {\em{algebraic polar form}} (APF) and represent the two ANFs by two graphs, where the polynomials,
$m''$ and $p''$, can be written as magnitude and phase graphs,
respectively. $p''$ maps to the {\em{phase graph}} $P''$ with vertex and edge sets
${\cal{V}}_{P''} = \{1,2\}$ and ${\cal{E}}_{P''} = \{12\}$,
respectively, and $m''$
maps to the {\em{magnitude graph}} $M''$ with vertex and edge sets
${\cal{V}}_{M''} = \{0,1\}$ and
${\cal{E}}_{M''} = \{01\}$, respectively.
The method of mapping a
magnitude polynomial, $m({\bf{x}})$ to its associated magnitude graph, $M$,
is explained in definition \ref{df:TGS}.
Although we are conceptually dealing with a two-graph object, $(M,P)$,
we prefer to act on an associated single graph, $G$, where the
vertex and edge sets of $G$ satisfy
${\cal{V}} = {\cal{V}}_M \cup {\cal{V}}_P$,
${\cal{E}} = {\cal{E}}_M \cup {\cal{E}}_P$, respectively.
If we further bipartition the vertex set ${\cal{V}}$ into ${\cal{L}}$ and ${\cal{R}}$,
where ${\cal{V}} = {\cal{L}} \cup {\cal{R}}$, ${\cal{L}} \cap {\cal{R}} = \emptyset$, and
${\cal{R}} = {\cal{V}}_P$, then we can exactly recover the
graph pair, $(M,P)$, from the graph-set pair, $(G,{\cal{R}})$, so the graph pair and graph-set pair definitions are equivalent.

\subsection{Local equivalence and a subgroup of the local Clifford group}
{\em{Measurement-based quantum computing}} using {\em{cluster states}} \cite{Raus:MBQC} or, more generally, graph
states, considers the action of unitary matrices on the graph state,
along with measurement of its vertices and classical communication
between its vertices. Of particular importance are the action of those unitaries
from ${\bf{C_n}}$ on the graph state \cite{Raus:MBQC}.
A classification of the equivalence classes of graph states,
wrt unitaries from
${\bf{C_n}}$, has been undertaken \cite{Hohn:SD,Hein:Gr,DanClass,DanDat,Glynn:Geom}, and, until very recently, it was an
open problem to prove that such equivalence classes remain the same even when one widens the class of unitaries considered to include
local unitaries outside the local Clifford group \cite{Sch:LULC}. Recent results have, however, suggested that this
so-called `LU$=$LC conjecture' is false \cite{Gro:LULC,Ji:LULC}.
Equivalence of graph states wrt the action of unitaries from ${\bf{C_n}}$ can be
realised on the associated graphs by means of {\em{local complementation}}
\cite{Bou:Tree,Bou:Grph,Glynn:Graph,VanD:Gr,DanQECC}. In \cite{RP:BC1} it was
shown that successive local complementations on a graph can be realised by considering the action on the graph state
of only a small subgroup, ${\bf{T_n}}$, of matrices from
${\bf{C_n}}$, where ${\bf{T_n}} = {\bf{T_1}}^{\otimes n}$ and
${\bf{T_1}} = \{I,\lambda,\lambda^2\}$ is a cyclic subgroup generated by
$\lambda = \frac{\omega^5}{\sqrt{2}} \begin{tiny} \left ( \begin{array}{rr}
1 & i \\
1 & -i
\end{array} \right ) \end{tiny}$, where
${\bf{T_1}} \subset {\bf{C_1}}$, and $|{\bf{T_1}}| = 3$. We call ${\bf{T_n}}$ the {\em{local transform group}} over $n$ qubits.
Moreover ${\bf{C_1}} = {\bf{T_1}} \times {\bf{D_1}}$,
where $|{\bf{C_1}}| = 192$ and $|{\bf{D_1}}| = 64$, and ${\bf{D_1}}$ is a
subgroup of diagonal and antidiagonal $2 \times 2$
matrices generated by $\omega$,
$\begin{tiny} \left ( \begin{array}{rr}
1 & 0 \\
0 & i
\end{array} \right ) \end{tiny}$, and
$\begin{tiny} \left ( \begin{array}{rr}
0 & 1 \\
1 & 0
\end{array} \right ) \end{tiny}$.
In \cite{RP:BC1} we concentrate only on the subset of transforms
from ${\bf{T_n}}$ whose action on a graph state yield flat
spectra, where these flat spectra can be interpreted, to within a final multiplication ny a member of ${\bf D_n}$, as a set of
locally-equivalent graph states.
In this paper we, more generally, consider the action of all $3^n$ transforms from
${\bf{T_n}}$ on a graph state. We show that
a graph state is always locally equivalent, wrt unitaries from ${\bf C_n}$,
to a two-graph object, $(M,P)$, where $M$ and $P$
represent {\em{magnitude}} and {\em{phase}} graphs for the state, respectively,
and the action of any member of ${\bf{T_n}}$ on such a state can be expressed as a
graphical
operation on the combined graph formed by $M$ and $P$, to yield another graph which can, once again,
be split into a two-graph, $(M',P')$ object.

To compute the two-graph orbit and/or perform spectral analysis of a certain
graph or stabilizer state, neither \cite{RP:BC1} or this paper use ${\bf{T_1}}$ explicitly.
Instead we use the set of three matrices, $\{I,H,N\}$.
It is evident that
$\lambda = \omega^5 N$, and $\lambda^2 = \omega^3
\begin{tiny} \left ( \begin{array}{rr}
1 & 0 \\
0 & -i
\end{array} \right ) \end{tiny}H$,
so one can always obtain the action of any unitary from the transform group, ${\bf{T_n}}$, by first applying the appropriate
unitary from $\{I,H,N\}^{\otimes n}$, then applying a suitable unitary from ${\bf{D_n}}$,
where $\{U_0,U_1,\ldots,U_{t-1}\}^{\otimes n}$ means the set of matrices formed by any
$n$-fold tensor product of matrices from the set $\{U_0,U_1,\ldots,U_{t-1}\}$.
But the application of any unitary from ${\bf{D_n}}$ to a state does not change coefficient magnitudes.
So, to perform spectral analysis based on magnitude computations,
we can use $\{I,H,N\}$ instead of ${\bf{T_1}} = \{I,\lambda,\lambda^2\}$.
We choose to do this because $H$ is the
2-point periodic {\em{discrete Fourier transform}} (DFT), and $N$ is the 2-point negaperiodic DFT, and using this viewpoint
facilitates a `Fourier' approach to the analysis of graph states and stabilizer (two-graph)
states. However, all results in this paper wrt
$\{I,H,N\}^{\otimes n}$ are trivially translated into results wrt ${\bf{T_n}}$, as shown in subsection \ref{lambda}.

\subsection{Main aims of this paper}
In previous work the use of graphs to represent graph states has simplified both
theoretical and computational analyses of graph
states. Our primary aim, in this paper, is to use two-graph states to represent
stabilizer states, so as to simplify analysis of the stabilizer state, where the
graph state is a special case of the two-graph state.
We obtain computationally efficient algorithms for the spectral analysis of the graph and
two-graph state wrt ${\bf{C_n}}$, as
the set of spectra computed via the action of $\{I,H,N\}^{\otimes n}$ on a two-graph state
acts as a precise summary of the much larger set of spectra resulting from the action of any
member of ${\bf{C_n}}$ on the two-graph state, where the action of ${\bf{D_n}}$ has been factored out.
A secondary aim of this paper is to provide an efficient, localised, graphical method to
realise the action of any member of ${\bf{C_n}}$
on the graph or two-graph state. This is made possible because $H$ and $N$ are generators of ${\bf{C_1}}$ and, in this paper, we characterise
the actions of $H$ and $N$ on the two-graph state and, therefore, ${\bf{C_n}}$ is covered via repeated actions of $H$ and $N$.
Moreover, as ${\bf C_n} = {\bf D_n} \times {\bf T_n}$, and ${\bf T_n} = D\{I,H,N\}^{\otimes n}$, $D \in {\bf D_n}$ then,
to within a final action by a member of ${\bf{D_n}}$, the
graphical characterisation of the action of any unitary from $\{I,H,N\}^{\otimes n}$ on a two-graph object,
$(G,{\cal{R}}) \equiv (M,P)$ will, at the same time, graphically
characterise the action of successive unitaries
from $\{I,H,N\}^{\otimes n}$ on $(G,{\cal{R}})$.

Section \ref{Prelim} onwards of this paper makes precise the discussion of this introduction.
Let $U_v = I^{\otimes v} \otimes U \otimes I^{\otimes n-v-1}$.
Then it is shown that
\bite
	\item {\em{Two-Graph State}}: The two-graph state
	comprises a graph with loops, $G$, and a set ${\cal{R}}$ or,
	equivalently, two graphs $M$ and $P$ ($(G,{\cal{R}}) \equiv (M,P)$), and is represented by $m(-1)^p$,
	where $m$ is a product of affine Boolean functions, and $p$ is a quadratic Boolean function,
	The transition between two representations of the same two-graph state
	is characterised via the operation called `swp' which operates on $(G,{\cal{R}})$.
	Then the action of a unitary, $H_v$, $v \in {\cal{V}}$, on $(G,{\cal{R}})$ is characterised
	via the conditional action of `swp' on $(G,{\cal{R}})$, and a set operation on ${\cal{R}}$,
	to produce another two-graph state, $(G',{\cal{R}}') \equiv (M',P')$.
	Consequently the action of any transform from $\{I,H\}^{\otimes n}$
	on a two-graph state can be computed graphically plus a few set operations.
	\item {\em{Generalised Two-Graph State}}: 
	The generalised two-graph state comprises a graph with loops,
	$G$, and two sets ${\cal{R}}$ and ${\cal{Q}}$ or, alternatively, two graphs
	$M$ and $P$ and a set ${\cal{Q}}$,
	($(G,{\cal{R}},{\cal{Q}}) \equiv (M,P,{\cal{Q}})$), and is represented by $mi^p$,
	where $m$ is a product of affine Boolean functions, and $p$ is a quadratic
	function from ${\mathbb{Z}}_2^n \rightarrow {\mathbb{Z}}_4$ of the special form.
	The possible loops at vertices in ${\cal{R}}$ are weighted according to elements
	in ${\cal Q}$.
	The transition between two representations of the same generalised two-graph state
	is characterised via the generalised operation called `swp' which now operates on $(G,{\cal{R}},{\cal{Q}})$.
	Then the actions of unitaries, $H_v$ and $N_v$, $v \in {\cal{V}}$, on $(G,{\cal{R}},{\cal{Q}})$ can be characterised
	via the conditional action of `swp' on $(G,{\cal{R}},{\cal{Q}})$, and certain other
	conditional operations on $G$,${\cal{R}}$, and ${\cal{Q}}$,
	to produce another generalised two-graph state,
	$(G',{\cal{R}}',{\cal{Q}}') \equiv (M',P',{\cal{Q}}')$.
	Consequently the action of any transform from $\{I,H,N\}^{\otimes n}$
	on a generalised two-graph state can be computed graphically plus a few set operations.
	\item {\em{Spectral Analysis of the Graph State}}:
	By considering $L_j$ norms of the graph state wrt the local Clifford group, we demonstrate the usefulness
	of the generalised two-graph representation to compute, efficiently, these norms.
\eite

We also generalise the graph operations of {\em{edge-local complementation}} (ELC) \cite{RP:Int,VanD:Pivot,Dan:Pivot}
and {\em{local complementation}} (LC) \cite{Bou:Tree,Bou:Grph,Glynn:Graph,VanD:Gr,DanQECC} to the two-graph operations,
{\em{edge-local complementation}}$^{\odot}$ (ELC$^{\odot}$)
and {\em{local complementation}}$^{\odot}$ (LC$^{\odot}$) which now take into account graph loops.

A recent paper \cite{Ell:CliffGraph}, independent to ours, also extends the graphical
notation to deal with the action of the local Clifford group on stabilizer states. \cite{Ell:CliffGraph}
also implicitly utilises a bipartite splitting of the graph (via `hollow' and `filled-in' nodes), and also
requires graph loops. \cite{Ell:CliffGraph} describes the action of $H$, $S$ and $Z$ on their graph, whereas we describe
the action of $H$ and $N$. Their model and our model must be equivalent in terms
of characterising the action of the local Clifford group on stabilizer states. However one can list some
differences in approach between the papers as follows. Firstly, \cite{Ell:CliffGraph} focusses, primarily, on modelling the
action of the local Clifford group. In contrast, we focus, primarily, on modelling the action of the local transform
group, ${\bf T_n}$, and/or $\{I,H,N\}^{\otimes n}$ as we are more interested in
evaluating spectral metrics for the graph state as efficiently as possible,
up to as many qubits as possible, although a secondary result of our work is that the action of the complete local Clifford
group is also modelled. Secondly, \cite{Ell:CliffGraph} implicitly considers the stabilizer state as a joint eigenstate, and
does not therefore have to consider an explicit basis for the state. In contrast, in our paper we consider an explicit computational
basis for the state, and this allows us to distinguish between magnitude and phase properties of the stabilizer state. This,
in turn, allows us to evaluate spectral metrics, associated with the graph state, with small effort.
Thirdly, by distinguishing between magnitude and phase, we highlight the stabilizer state as a simultaneous
generalisation of both the usual classical cryptographic representation of Boolean functions (the phase part), and
the usual parity-check graph (factor graph) representation of classical binary linear codes (the magnitude part). The link to parity-check
graphs was investigated in \cite{Par:QFG} and the interaction between magnitude and phase graphs
was investigated in \cite{Par:QE} and has since been exploited in \cite{Par:SB,RP:BC1,RP:Piv,DanQECC,DanAPC,DanPAR,DanClass}.
A preliminary version of this paper was presented at \cite{graphOxford}.

For the rest of this paper we only consider
connected graph states as, otherwise, the system is degenerate.
We also ignore the global multiplicative constants in front of the
state vector. In particular
our method strictly only distinguishes between the action on the two-graph state
of matrices from the size $24^n$ subgroup of the local Clifford group,
as the supplementary multiplication of
the state by a power of $\omega$ is ignored, i.e we remove the centre of the local Clifford
group. For most scenarios this global multiplicative constant can be
ignored, however a trivial refinement of our method would be necessary if one was to relate the action of the same sequence
of matrices from the local Clifford
group on two or more two-graph states.

\section{Formal Definitions}
\label{Prelim}

Define $U_v = I^{\otimes v} \otimes U \otimes I^{\otimes n-v-1}$.

\begin{df} 
Let $I= \begin{tiny} \left(\begin{array}{cc}
1&0\\
0&1 
\end{array}\right) \end{tiny}$, $H=\frac{1}{\sqrt{2}} \begin{tiny}\left(\begin{array}{cc}
1&1\\
1&-1 
\end{array}\right) \end{tiny}$, and $N=\frac{1}{\sqrt{2}} \begin{tiny}\left(\begin{array}{cc}
1&i\\
1&-i 
\end{array}\right) \end{tiny}$ be the $2\times2$ identity,
{\em Walsh-Hadamard}, and {\em negahadamard} \cite{Par:SB} matrices, respectively.
The set of $3^n$ transforms, $\{I,H,N\}^{\otimes n}$, is defined as the set of all $n$-fold tensor product
combinations of matrices $I,\,H$ and $N$.
 \label{IHN:def}\end{df}

\begin{df} \cite{Par:QE}
A pure $n$-qubit state,
$\left| \psi \right > = (\left| \psi \right >_{0\ldots 00},\left| \psi \right >_{0\ldots 01},
\ldots,\left| \psi \right >_{1\ldots 11})$, with vector entries satisfying
$\left| \psi \right >_{\bf{x}} \in
c\{0,1,-1\}$, for some complex constant, $c$, can always be written in the form
$$ \left| \psi \right > = cm({\bf x})(-1)^{p({\bf x})}, $$
where $\left| \psi \right >_{\bf{x}} = cm({\bf x})(-1)^{p({\bf x})}$,
$\forall {\bf{x}} \in {\mathbb{Z}}_2^n$,
and $m, p: {\mathbb{Z}}_2^n \rightarrow {\mathbb{Z}}_2$ are both Boolean functions.
The output of $m$ is embedded in the complex numbers.
We separate, thus, {\em magnitude}, $m$, and {\em phase}, $p$, of $\left| \psi \right >$, and call
such a representation the {\em algebraic polar form (APF)} of $\left| \psi \right >$.
\end{df}
{\bf{Remark: }} In order to simplify notation we henceforth omit the normalisation constant, $c$, from
any expression of the form $\left| \psi \right > = cm({\bf x})(-1)^{p({\bf x})}$ or similar.

\begin{df}
Let $G = ({\cal{V}},{\cal{E}})$ be a graph with vertex set, ${\cal{V}}$,
and edge set, ${\cal{E}}$, where $G$ may contain loops.
Let $\Gamma_G$ be the
binary adjacency matrix of $G$. Then, for two graphs, $G$ and $G'$,
both defined over the same $n$ vertices, $G'' = G \pm G'$ means that the adjacency matrix,
$\Gamma_{G''}$, of $G''$, satisfies
$\Gamma_{G''} = \Gamma_G \pm \Gamma_{G'}$.
Let ${\cal{N}}_v^G$ be the set of vertices other than $v$ which are neighbours of vertex $v$
in $G$. 
Let ${\cal{B}}_v^G = {\cal{N}}_v^G \cup \{v\}$ be the set of vertices
less than or equal to one edge distance from vertex $v$ in $G$.
For a vertex set, ${\cal{V}}$, let $G_{\cal{V}}$ be the {\em{induced subgraph}} of $G$ on ${\cal{V}}$,
comprising all edges from $G$ whose endpoints are both in ${\cal{V}}$.
For vertex sets, ${\cal{V}}$ and ${\cal{V}}'$, define
$K_{{\cal{V}},{\cal{V}}'}$ to be the graph with binary adjacency matrix,
$\Gamma_K$, where $\Gamma_{K_{ij}} = \Gamma_{K_{ji}} = 1$ iff
$i \in {\cal{V}}\setminus{\cal{V}}', j \in {\cal{V}}'$
or $i=j\in{\cal{V}}\cap{\cal{V}}'$.
$G_{\cal{V}}$ and $K_{{\cal{V}},{\cal{V}}'}$ may contain loops. Let $G_v = K_{\{v\},{\cal{N}}_v^G}$.
Let $\Delta_{\cal{V}}$ be the graph with diagonal binary adjacency matrix,
$\Gamma_{\Delta}$, where $\Gamma_{\Delta_{ij}} = 1$ iff $i = j \in {\cal{V}}$.
The {\em{complete graph}}, $C_{{\cal{V}}}$, is the simple graph whose
edge set comprises the set of edges $\{vw, \mz
\forall v,w \in {\cal{V}}, v < w \}$.
\label{graphdefs} 
\end{df}

\begin{df} \cite{Bou:Tree,Bou:Grph,Glynn:Graph,VanD:Gr,DanQECC}
The action of {\em{local complementation}} (LC) on a simple graph $G$ at vertex $v$ is
the graph transformation obtained by replacing the subgraph $G_{{\cal{N}}_v^G}$
by its complement.
\end{df}
Example: The action of LC on a graph at vertex $v=0$, is shown in figure \ref{LCact}.
\begin{figure}[!h]
\cb
\includegraphics[width=.5\hsize]{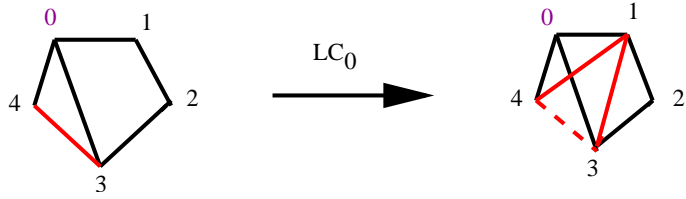}
\ce
\caption{The action of LC at vertex $0$}
\label{LCact}
\end{figure}

\begin{df} \cite{RP:Int,VanD:Pivot,Dan:Pivot}
The action of {\em{edge local complementation}} (ELC) on a simple graph $G$ at edge $vw$ is
the graph transformation obtained by performing LC at vertex $v$, then vertex $w$, then
vertex $v$ again (or, equivalently, at $w$, then $v$, then $w$ again).
\end{df}

In this paper we generalise both LC and ELC so as to operate on a two-graph object.

\begin{df}
Let $G$ be a graph with possible loops, containing an edge $vw$, $v \ne w$.
Then $G^{vw}$ is the graph resulting from
the action of {\em{edge local complementation}}$^{\odot}$ (ELC$^{\odot}$) on edge $vw$ of $G$, where
$$ G^{vw} = G + K_{{\cal{B}}_v^G,{\cal{B}}_w^G} + \Delta_{\{v,w\}}
 + \Gamma_{G_{vv}}\Delta_{{\cal{B}}_w^G} + \Gamma_{G_{ww}}\Delta_{{\cal{B}}_v^G}\enspace. $$
\label{ELC}
\end{df}

\noindent
Example: The action of ELC$^{\odot}$ on the following graph at edge $vw=31$, is shown in figure \ref{ELCLoop}.
\begin{figure}[!h]
\cb
\includegraphics[width=.5\hsize]{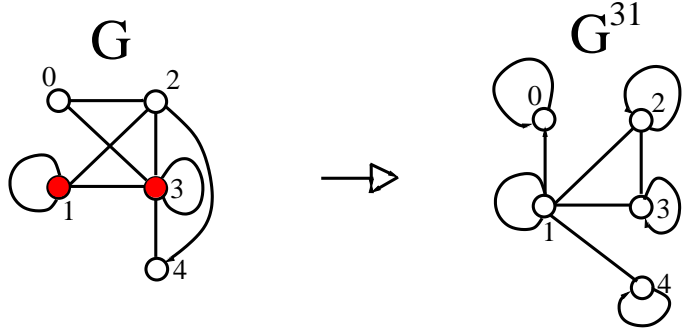}
\ce
\caption{The action of ELC$^{\odot}$ at edge $31$}
\label{ELCLoop}
\end{figure}

\noindent {\bf{Remark: }} From definition \ref{ELC}, even when $G$ is a simple graph, $\Gamma_{G_{vv}} = \Gamma_{G_{ww}} = 0$,
we see that possible loops can still be produced from term $K_{{\cal{B}}_v^G,{\cal{B}}_w^G}$.
The ELC
operation, which acts only on simple graphs, can be recovered from ELC$^{\odot}$ by applying ELC$^{\odot}$ to a simple graph, then
deleting any resultant loops from the output. 

\vspace{3mm}

The {\em{Pauli matrix group}} is generated by
$X = \begin{tiny} \left ( \begin{array}{rr}
0 & 1 \\
1 & 0
\end{array} \right ) \end{tiny}$,
$Z = \begin{tiny} \left ( \begin{array}{rr}
1 & 0 \\
0 & -1
\end{array} \right ) \end{tiny}$,
and $i$.
Let $S = \begin{tiny} \left ( \begin{array}{rr}
1 & 0 \\
0 & i
\end{array} \right ) \end{tiny}$.

\begin{df}
The {\em{local Clifford matrix group}}, ${\bf{C_1}}$,
is the group of 192 $2 \times 2$ matrices that normalise the Pauli group,
and can be decomposed as ${\bf{C_1}} = {\bf{D_1}} \times {\bf{T_1}}$ where
we call ${\bf{D_1}}$ the {\em{diagonal group}} and ${\bf{T_1}}$ the {\em{transform group}}.
${\bf{T_1}} = \{I,\lambda,\lambda^2\}$ is a cyclic subgroup generated by
$\lambda = \omega^5N$, where $i = \sqrt{-1}$ and $\omega = \sqrt{i}$,
${\bf{D_1}} = {\bf{C_1}} / {\bf{T_1}}=\left< S,X,\omega \right>$,
and comprises only diagonal or antidiagonal $2 \times 2$
matrices, and
$|{\bf{C_1}}| = 192$, $|{\bf{T_1}}| = 3$, and $|{\bf{D_1}}| = 64$.
We call ${\bf{C_n}}$, ${\bf{T_n}}$, and ${\bf{D_n}}$, the groups formed by $n$-fold tensor
products of matrices from ${\bf{C_1}}$, ${\bf{T_1}}$, and ${\bf{D_1}}$, respectively, where
$|{\bf C_n}| = 8 \times 24^n$ and $|{\bf D_n}| = 8 \times 8^n$.
\label{groups}
\end{df}

Observe that $\lambda = \omega^5N$ and
$\lambda^2 = \omega^3S^{-1}H$ so, for any $U \in {\bf{T_n}}$, and any
$V \in \{I,H,N\}^{\otimes n}$, we have
$U = \omega DV$ for some $D \in {\bf{D_n}}$. For the rest of the paper we focus on the action of
$\{I,H,N\}^{\otimes n}$ on the graph state and, more generally, on the two-graph state, where the alternative action
of ${\bf{T_n}}$ on the state can be derived easily (see section \ref{lambda}).

\begin{df} \cite{Raus:QC}
Given a graph, $P$, on $n$ vertices with adjacency matrix, $\Gamma_P$, define $n$ commuting Pauli operators
$$ {\cal{K}}_{P_j} = X_j\prod_{k \in {\cal{N}}_j} Z_k = X_j\prod_{k=0}^{n-1} Z_k^{\Gamma_{P_{jk}}}, $$
where ${\cal{N}}_j$ is the set of vertices in $P$ that are neighbours of vertex $j$.
The stabilizer, ${\cal{K}}_P$, is generated by $\langle{\cal{K}}_{P_0},{\cal{K}}_{P_1},\ldots,{\cal{K}}_{P_{n-1}}\rangle$, and $\left| \psi \right >$
is a {\em{graph state}} iff
$ {\cal{K}}_P\left| \psi \right > = \left| \psi \right >$,
for some simple graph, $P$. Explicitly,
in the computational basis, \cite{RP:BC1,VanD:Gr},
$$ \left| \psi \right >_{\bf x} = (-1)^{\sum_{i < j} \Gamma_{P_{ij}}x_ix_j}. $$
Any state
$\left| \psi \right >' = U\left| \psi \right >$, $U \in {\bf{C_n}}$,
is a {\em{stabilizer state}} locally equivalent to $\left| \psi \right >$.
\label{graphstate}
\end{df}

\section{The Two-Graph State}
\label{TGS}

\begin{df}
A {\em{two-graph state}} is a pure quantum state, $\left| \psi \right >$, of $n$ qubits that
can be defined by a graph, $G = ({\cal{V}},{\cal{E}})$, and a bipartition, $({\cal{L}},{\cal{R}})$,
where ${\cal{V}} = {\cal{L}} \cup {\cal{R}}$ and ${\cal{L}} \cap {\cal{R}} = \emptyset$, and where
$G_{\cal{L}}$ is the empty graph apart from possible loops. The pair, $(G,{\cal{R}})$, explicitly encodes a
two-graph object, $(M,P)$, where $P = G_{\cal{R}}$, and $M = G - P$ is a bipartite graph.
The state, $\left| \psi \right >$,
is defined by its algebraic polar form, $\left| \psi \right > = cm({\bf{x}})(-1)^{p({\bf{x}})}$, where
$c\in{\mathbb C}$, $m({\bf{x}}) : {\mathbb{Z}}_2^n \rightarrow {\mathbb{Z}}_2$ is a product of affine functions of the form,
$$ m({\bf{x}}) = \prod_{i \in {\cal{L}}} (\Gamma_{M_{ii}} + 1 + x_i + \sum_{j \in {\cal{R}}} \Gamma_{M_{ij}}x_j), $$
such that $m = 1$ when ${\cal{L}} = \emptyset$, and where
$p({\bf{x}}) : {\mathbb{Z}}_2^n \rightarrow {\mathbb{Z}}_2$ is a
quadratic function of the form,
$$ p({\bf{x}}) = \sum_{i,j \in {\cal{R}}, i < j} \Gamma_{P_{ij}}x_ix_j
 + \sum_{j \in {\cal{R}}} \Gamma_{P_{jj}}x_j. $$
\label{df:TGS}
\end{df}

\noindent
{\bf{Remark: }} For $(G,{\cal{R}}) \equiv (M,P)$ a two-graph state, $M$ and $P$ cannot contain loops at vertices in 
${\cal{R}}$ and ${\cal{L}}$, respectively. Also,  although at first it seems that we don't distinguish between for instance 
$m=(x_0+x_1+x_2+1)$ and $m=(x_0+x_1+1)(x_0+x_2+1)$, we do: by definition \ref{df:TGS}, the form 
$m=(x_0+x_1+x_2+1)$ can be represented, non-uniquely, by ${\cal L}=\{0\}$ and 
${\cal R}=\{1,2\}$, while the form
$m = (x_0+x_1+1)(x_0+x_2+1)$ can be represented, non-uniquely, by ${\cal L}=\{1,2\}$ and 
${\cal R}=\{0\}$. The factorization of $m$ into a product of affine terms of the form shown
in definition \ref{df:TGS} reflects the fact that $m$ represents a binary
linear coset code, ${\cal C}$, where each affine factor of $m$ represents a row of a systematic parity check matrix,
${\cal H}$, for ${\cal C}$, where ${\cal L}$ is an information set for ${\cal C}$.
For instance, with ${\cal R} = \{0,1,4\}$,
$m = (x_2 + x_0 + x_1 + 1)(x_3 + x_1 + x_4)(x_5 + x_0  + x_4 + 1)$ represents the systematic parity check matrix,
$ {\cal H} = \left ( \begin{tiny} \begin{array}{cccccc}
1 & 1 & 1 & 0 & 0 & 0 \\
0 & 1 & 0 & 1 & 1 & 0 \\
1 & 0 & 0 & 0 & 1 & 1
\end{array} \end{tiny} \right ), $ for the binary linear coset code, ${\cal C}$, with coset leader $000100$.

\indent 

We first describe the action of `$\m{swp}$' on the two-graph state at edge $vw$.
\begin{df} Let $\left| \psi \right > = m(-1)^p$ be a two-graph state over $n$ qubits, represented by the graph-set
$(G,{\cal{R}}) \equiv (M,P)$.
Let $v \in {\cal{L}}$ and $w \in {\cal{N}}_v^G$.
Then the action
of $\m{swp}$ at edge $vw$ is the operation that interchanges the roles of $v$ and $w$; i.e. the operation that
takes ${\cal R}$ to ${\cal{R}}' = {\cal{R}} \cup \{v\} \setminus \{w\}$, and results in a two-graph state, 
$m(-1)^{p'}$, where
$\left| \psi \right > = m(-1)^{p'} = m(-1)^p$.
\end{df}

\noindent {\bf{Remark:}}
The action of `swp' does not change $\left| \psi \right >$ or $m$, but it changes the
graphical representation $(G,{\cal{R}}) \equiv (M,P)$ to $(G',{\cal{R}}') \equiv (M',P')$.
In coding-theoretic terms, `swp' at $vw$ updates the information set, ${\cal L}$, to
${\cal L}'$, corresponding to an update of the systematic parity-check matrix for the code,
${\cal C}$, represented by $M$. The update in parity-check matrix induces a corresponding
modification of $P$ to $P'$.

\begin{lem}
Let $\left| \psi \right >$ be a two-graph state over $n$ qubits, represented by the graph-set
$(G,{\cal{R}}) \equiv (M,P)$.
Let $v \in {\cal{L}}$ and $w \in {\cal{N}}_v^G$.
Then the action
of $\m{swp}$ at edge $vw$ results in the two-graph state with associated graph-set, $(G',{\cal{R}}') \equiv (M',P')$,
which is obtained from $(G,{\cal{R}}) \equiv (M,P)$ as follows: \\

\noindent $(G',{\cal{R}}') = \m{\bf{swp}}(G,{\cal{R}},v,w)$:

\vspace{3mm}

$$ \begin{array}{l}
\left \{
\begin{array}{l}
{\cal{R}}' = {\cal{R}} \cup \{v\} \setminus \{w\}, \\
G' = G^{vw}.
\end{array}
\right .
\end{array} $$
\label{SwpAction}
\end{lem}
\begin{proof}
Section \ref{proofs}.
\end{proof}

We now describe the action of $H_v$ on a two-graph state.

\begin{thm}
Let $\left| \psi \right >$ be a two-graph state over $n$ qubits, represented by the graph-set
$(G,{\cal{R}}) \equiv (M,P)$.
Let $v \in \{0,1,\ldots,n-1\}$. 
Then $\left| \psi' \right > = H_v\left| \psi \right >$ is also a two-graph state
and can be described by the graph-set $(G',{\cal{R}}') \equiv (M',P')$, where \\

\noindent $(G',{\cal{R}}') = H_v(G,{\cal{R}})$:

\vspace{3mm}

$$
\begin{array}{ll}
\left \{
\begin{array}{l}
{\cal{R}}' = {\cal{R}} \cup \{v\}, \\
G' = G,
\end{array} \right . & \mf \mbox{ if }v \in {\cal{L}}, \\ \\
\left \{
\begin{array}{l}
{\cal{R}}' = {\cal{R}} \setminus \{v\}, \\
G' = G,
\end{array} \right . & \mf \mbox{ if }{\cal{B}}_v^G \subseteq {\cal{R}} \\ \\
\left \{
\begin{array}{l}
\m{assign } w  \in {\cal{N}}_v^M, \\
(G',{\cal{R}}'') = \m{swp}(G,{\cal{R}},w,v), \\
{\cal{R}}' = {\cal{R}}'' \cup \{v\},
\end{array} \right . & \mf \mbox{ if }v \in {\cal{R}}, {\cal{B}}_v^G \not\nsubseteq {\cal{R}}.
\end{array} $$
\label{TGSAction}
\end{thm}
\begin{proof}
Section \ref{proofs}.
\end{proof}

\noindent {\bf Example:} Let $\left|\psi \right >=m(-1)^p$ be a two-graph state, with $n=5$,
$m=(x_0+x_2+x_3+1)(x_1+x_2+x_3)$, $p=x_2x_3 + x_2x_4 + x_3x_4 + x_3$, and
graph $(G,{\cal{R}}) \equiv (M,P)$, where $G$ has edge set
${\cal{E}} = \{02,03,12,13,23,24,34,11,33\}$ and ${\cal{R}} = \{2,3,4\}$.
Then the action of $H_3$ on $\left|\psi \right >$ can be detailed as follows. Observe that
${\cal{B}}_3^G = \{0,1,2,3,4\} \not\subset {\cal{R}}$. Therefore, from
theorem \ref{TGSAction}, we can, arbitrarily, choose $w = 1$, as $1 \in {\cal{N}}_3^M$.
Then $(G',{\cal{R}}'') = \m{swp}(G,{\cal{R}},1,3)$, where
$G'$ has edge set ${\cal{E}} = \{01,12,13,14,23,00,11,22,33,44\}$ and
${\cal{R}}'' = \{1,2,4\}$. Finally we update ${\cal{R}}''$ to obtain
${\cal{R}}' = \{1,2,3,4\}$. The resulting graph, $(G',{\cal{R}}') \equiv (M',P')$,
represents the two-graph state $\left|\psi' \right >=m'(-1)^{p'}$, where
$m' = (x_0 + x_1)$ and $p' = x_1x_2 + x_1x_3+ x_1x_4 + x_2x_3 + x_1 + x_2 + x_3 + x_4$.
This example is illustrated in figure \ref{figHv}.

\begin{figure}[!h]
\cb
\includegraphics[width=5in]{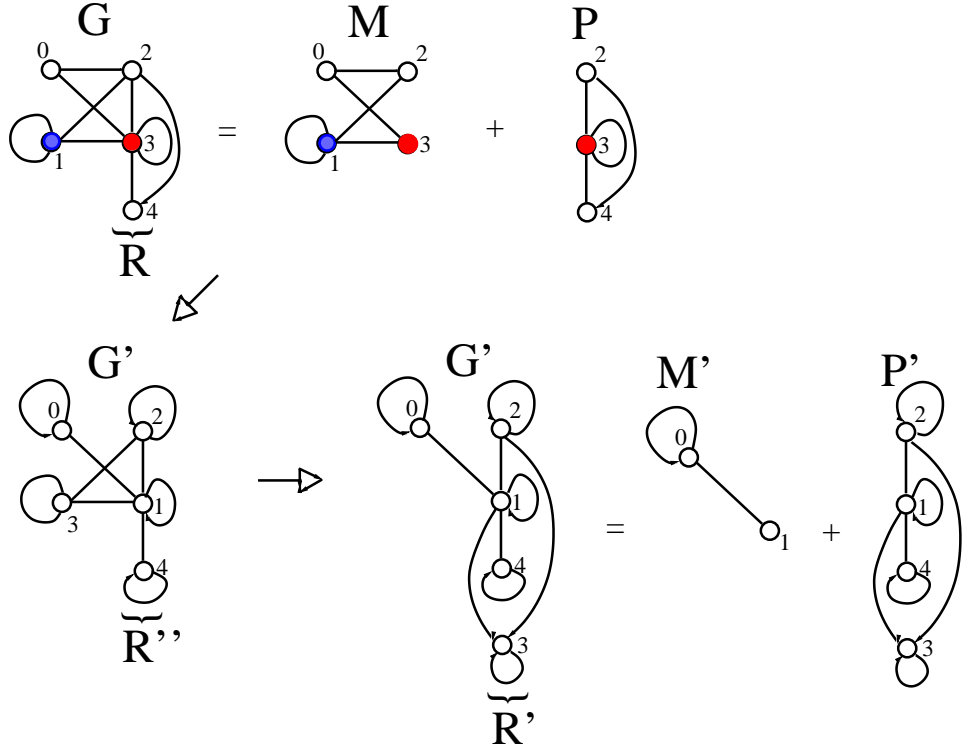}
\ce
\caption{The action of $H_3$ on a two-graph state}
\label{figHv}
\end{figure}


\begin{thm} Let $\left| \psi \right > =m(-1)^p$ be a two-graph state over $n$ qubits. Then
there always exists a graph state, $\left| \psi \right >'$, such that
$\left| \psi \right >' = DU\left| \psi \right >$, where $U \in \{I,H\}^{\otimes n}$,
and $D \in {\bf{D_n}}$.
\label{twographeq}
\end{thm}
\begin{proof}
Select an arbitrary $v \in {\cal{L}}$, and apply $H_v$ to $\left| \psi \right >$. Then, by applying the algorithm
of theorem \ref{TGSAction}, we obtain
$\left| \psi \right >^{(1)} = H_v\left| \psi \right >$, where ${\cal{L}}^{(1)} = {\cal{L}} \setminus \{v\}$.
Select an arbitrary $v' \in {\cal{L}}^{(1)}$ and repeat the above process by applying $H_{v'}$ to
$\left| \psi \right >^{(1)}$ so as to obtain $\left| \psi \right >^{(2)}$, and so on.
After $k = |{\cal{L}}|$ such recursions one
obtains ${\cal{L}}^{(k)} = \emptyset$, which implies that $\left| \psi \right >^{(k)}$ is a graph state to
within loops in $P$, as $m = 1$. The loops in $P$ can then be eliminated via the action of matrices from ${\bf{D_n}}$.
\end{proof}

\begin{cor} (of theorem \ref{twographeq})
\label{corthm2}
The two-graph state is a stabilizer state.
\end{cor}
\begin{proof}
It is known that a stabilizer state is locally-equivalent to a graph state \cite{Sch:QG,Gras:QG}, and local-equivalence is
reversible.
\end{proof}

\section{The Generalised Two-Graph State}\label{TGSgen}

For a set of integers, ${\cal{Q}}$, let ${\cal{Q}}_j = 1$ iff $j \in {\cal{Q}}$, otherwise ${\cal{Q}}_j = 0$.

\begin{df}
A {\em{generalised two-graph state}} is a pure quantum state, $\left| \psi \right >$, of $n$ qubits that
can be defined by the graph-set-set, $(G,{\cal{R}},{\cal{Q}})$,
where $G = ({\cal{V}},{\cal{E}})$ is an $n$-vertex graph,
with bipartition, $({\cal{L}},{\cal{R}})$,
where ${\cal{V}} = {\cal{L}} \cup {\cal{R}}$ and ${\cal{L}} \cap {\cal{R}} = \emptyset$, where
$G_{\cal{L}}$ is the empty graph apart from possible loops, and where ${\cal{Q}} \subset {\cal{R}}$.
The triple, $(G,{\cal{R}},{\cal Q})$, explicitly encodes a generalised
two-graph, $(M,P,{\cal Q})$, where $P = G_{\cal{R},{\cal Q}}$, and $M = G - P$ is bipartite.
The state, $\left| \psi \right >$,
is defined by its algebraic polar form, $\left| \psi \right > = cm({\bf{x}})i^{p({\bf{x}})}$, where
$c\in{\mathbb C}$, $m({\bf{x}}) : {\mathbb{Z}}_2^n \rightarrow {\mathbb{Z}}_2$ is a product of affine functions of the form,
$$ m({\bf{x}}) = \prod_{i \in {\cal{L}}} (\Gamma_{M_{ii}} + 1 + x_i + \sum_{j \in {\cal{R}}}
\Gamma_{M_{ij}}x_j), $$
such that $m = 1$ when ${\cal{L}} = \emptyset$, and where
$p({\bf{x}}) : {\mathbb{Z}}_2^n \rightarrow {\mathbb{Z}}_4$ is a
quadratic function of the form,
$$ p({\bf{x}}) = \sum_{i,j \in {\cal{R}}, i < j} 2\Gamma_{P_{ij}}x_ix_j
 + \sum_{j \in {\cal{R}}} (2\Gamma_{P_{jj}} + {\cal{Q}}_j)x_j. $$
\label{TGS1}
\end{df}

\noindent
{\bf Remark}: The generalised two-graph state can alternatively and, perhaps, more naturally,
be viewed as a graph with weighted ${\mathbb{Z}}_4$
loops and a set ${\cal{R}}$. But we choose the equivalent $(G,{\cal{R}},{\cal{Q}})$ representation for notational convenience.
When ${\cal{Q}} = \emptyset$, then the generalised two-graph state,
defined by $(G,{\cal{R}},{\cal{Q}})$,
reduces to the two-graph state, defined by $(G,{\cal{R}})$,
and non-empty ${\cal{Q}}$ introduces linear
terms over ${\mathbb{Z}}_4$ to the state.

Let $A \ominus B$ be the symmetric difference of sets $A$ and $B$, that is $A \ominus B=(A\setminus B)\cup
(B\setminus A)$.

\begin{df}
Let $(G,{\cal{Q}})$ be the graph-set pair, extracted from the generalised two-graph state $(G,{\cal{R}},{\cal{Q}})$,
with $G$ an $n$-vertex graph with possible loops and
${\cal{Q}} \subset \{0,1,\ldots,n-1\}$.
Then $(G,{\cal{Q}})^v = (G^v,{\cal{Q}}^v)$ is defined to be the graph-set pair resulting from
the action of {\em{local complementation}}$^{\odot}$ (LC$^{\odot}$) on vertex $v$ of $(G,{\cal{Q}})$, where
$$ \begin{array}{l}
G^{v} = G + C_{{\cal{N}}_v^G} + \Gamma_{G_{vv}}\Delta_{{\cal{N}}_v^G} + \Delta_{{\cal{Q}} \cap {\cal{N}}_v^G}, \\
{\cal{Q}}^v = {\cal{Q}} \ominus {{\cal{B}}_v^G}.
\end{array} $$
\label{LC}
\end{df}

\noindent
{\bf Example: } The action of LC$^{\odot}$ on the following graph-set, $(G,{\cal{Q}})$, at vertex $v=3$, is shown in figure \ref{LCLoop}.
\begin{figure}[!h]
\cb
\includegraphics[width=.5\hsize]{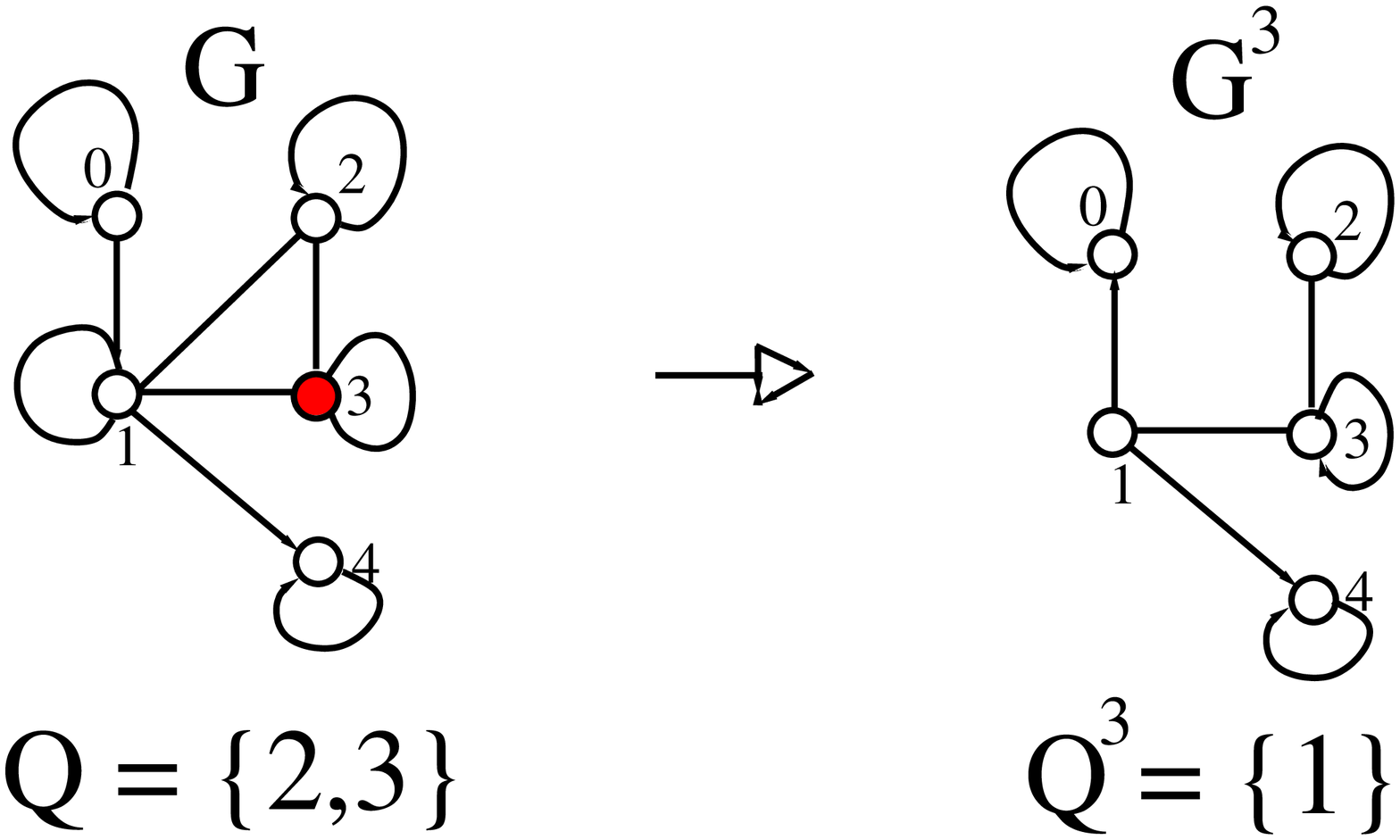}
\ce
\caption{The action of LC$^{\odot}$ at vertex $3$}
\label{LCLoop}
\end{figure}

\noindent
{\bf{Remark: }} From definition \ref{LC}, even when $G$ is a simple graph and ${\cal{Q}} = \emptyset$,
we see that possible loops can still be produced at the output. The LC
operation, which acts only on simple graphs, can be recovered from LC$^{\odot}$ by applying LC$^{\odot}$ to a simple graph, then
deleting any resultant loops from the output.

We now describe the action of `$\m{swp}$' on the generalised two-graph state at edge $vw$, as a natural
extension of `swp' on a two-graph state.

\begin{df} Let $\left| \psi \right >  = mi^p$ be a generalised two-graph state over $n$ qubits, represented by the graph-set-set
$(G,{\cal{R}},{\cal{Q}}) \equiv (M,P,{\cal{Q}})$. Let $v \in {\cal{L}}$ and $w \in {\cal{N}}_v^G$.
Then the action
of $\m{swp}$ at edge $vw$ is the operation that interchanges the roles of $v$ and $w$; i.e. the operation that
takes ${\cal R}$ to ${\cal{R}}' = {\cal{R}} \cup \{v\} \setminus \{w\}$, and results in a generalised two-graph state,
$m'i^{p'}$, where $\left| \psi \right > = m'i^{p'} = mi^p$.
\end{df}

\noindent {\bf{Remark:}} `swp' does not change $\left| \psi \right >$.

\begin{lem}
Let $\left| \psi \right >$ be a generalised two-graph state over $n$ qubits, represented by the graph-set-set
$(G,{\cal{R}},{\cal{Q}}) \equiv (M,P,{\cal{Q}})$.
Let $v \in {\cal{L}}$ and $w \in {\cal{N}}_v^G$. 
Then the action
of $\m{swp}$ at edge $vw$ results in the generalised two-graph state with associated graph-set-set,
$(G',{\cal{R}}',{\cal{Q}}') \equiv (M',P',{\cal{Q}}')$,
and is obtained from $(G,{\cal{R}},{\cal{Q}})$ as follows: \\

\noindent Let $v \in {\cal{L}}, w \in {\cal{N}}_v^G$. Then, $(G',{\cal{R}}',{\cal{Q}}') = 
\m{\bf{swp}}(G,{\cal{R}},{\cal{Q}},v,w)$ can be expressed as:

\vspace{3mm}

$$ \begin{array}{l}
\left \{
\begin{array}{l}
{\cal{R}}' = {\cal{R}} \cup \{v\} \setminus \{w\}, \\
G'' = G^{vw}, \\
\m{if } {\cal{Q}}_w = 1 \\
\mf (G',{\cal{Q}}') = (G'',{\cal{Q}})^w \\
\m{else } \\
\mf (G',{\cal{Q}}') = (G'',{\cal{Q}}).
\end{array} \right .
\end{array} $$
\label{SwpActiongen}
\end{lem}
\begin{proof}
Section \ref{proofs}.
\end{proof}

We now describe the action of $H_v$ on a generalised two-graph state.

\begin{thm}
Let $\left| \psi \right >$ be a generalised two-graph state over $n$ qubits, represented by the graph-set-set
$(G,{\cal{R}},{\cal{Q}}) \equiv (M,P,{\cal{Q}})$.
Let $v \in \{0,1,\ldots,n-1\}$. 
Then $\left| \psi' \right > = H_v\left| \psi \right >$ is also a generalised two-graph state
and can be described by the graph-set-set $(G',{\cal{R}}',{\cal{Q}}') \equiv (M',P',{\cal{Q}}')$, where \\

\noindent $(G',{\cal{R}}',{\cal{Q}}') = H_v(G,{\cal{R}},{\cal{Q}})$:

\vspace{3mm}

$$
\begin{array}{ll}
\left \{
\begin{array}{l}
{\cal{R}}' = {\cal{R}} \cup \{v\}, \\
(G',{\cal{Q}}') = (G,{\cal{Q}}),
\end{array} \right . & \mf \mbox{ if }v \in {\cal{L}}, \\ \\
\left \{
\begin{array}{l}
\m{if } {\cal{Q}}_v = 0 \\
\mf {\cal{R}}' = {\cal{R}} \setminus \{v\}, \\
\mf (G',{\cal{Q}}') = (G,{\cal{Q}}), \\
\m{else } \\
\mf {\cal{R}}' = {\cal{R}}, \\
\mf (G'',{\cal{Q}}'') = (G,{\cal{Q}})^v\\
\mf G' = G'' + \Delta_{{\cal{B}}_v^G}\\
\mf {\cal Q}'={\cal Q}''\cup\{v\} \enspace,
\end{array} \right . & \mf \mbox{ if }{\cal{B}}_v^G \subseteq {\cal{R}} \\ \\
\left \{
\begin{array}{l}
\m{assign } w  \in {\cal{N}}_v^M, \\
(G',{\cal{R}}'',{\cal{Q}}') = \m{swp}(G,{\cal{R}},{\cal{Q}},w,v), \\
{\cal{R}}' = {\cal{R}}'' \cup \{v\},
\end{array} \right . & \mf \mbox{ if }v \in {\cal{R}}, {\cal{B}}_v^G \not\nsubseteq {\cal{R}}.
\end{array} $$
\label{HvAction}
\end{thm}
\begin{proof}
Section \ref{proofs}.
\end{proof}

\noindent {\bf Example:} Let $\left|\psi \right >=mi^p$ be a generalised two-graph state, with $n=5$,
$m=(x_0+x_2+x_3+1)(x_1+x_2+x_3)$, $p=2x_2x_3 + 2x_2x_4 + 2x_3x_4 + x_2 + 3x_3$, and
graph $(G,{\cal{R}},{\cal{Q}}) \equiv (M,P,{\cal{Q}})$, where $G$ has edge set
${\cal{E}} = \{02,03,12,13,23,24,34,11,33\}$, ${\cal{R}} = \{2,3,4\}$, and ${\cal{Q}} = \{2,3\}$.
Then the action of $H_3$ on $\left|\psi \right >$ can be detailed as follows where
we, arbitrarily, choose $w = 1$.
Then ${\cal{R}}' = \{1,2,3,4\}$ and ${\cal{Q}}' = \{1,2,3,4\}$. The resulting graph,
$(G',{\cal{R}}',{\cal{Q}}') \equiv (M',P',Q')$,
represents the generalised two-graph state $\left|\psi' \right >=m'i^{p'}$, where
$m' = (x_0 + x_1)$ and $p' = 2x_1x_3 + 2x_1x_4 + 2x_2x_3 + x_2 + x_3 + x_4$.
This example is illustrated in figure \ref{figHvQ}.

\begin{figure}[!h]
\cb
\includegraphics[width=5in]{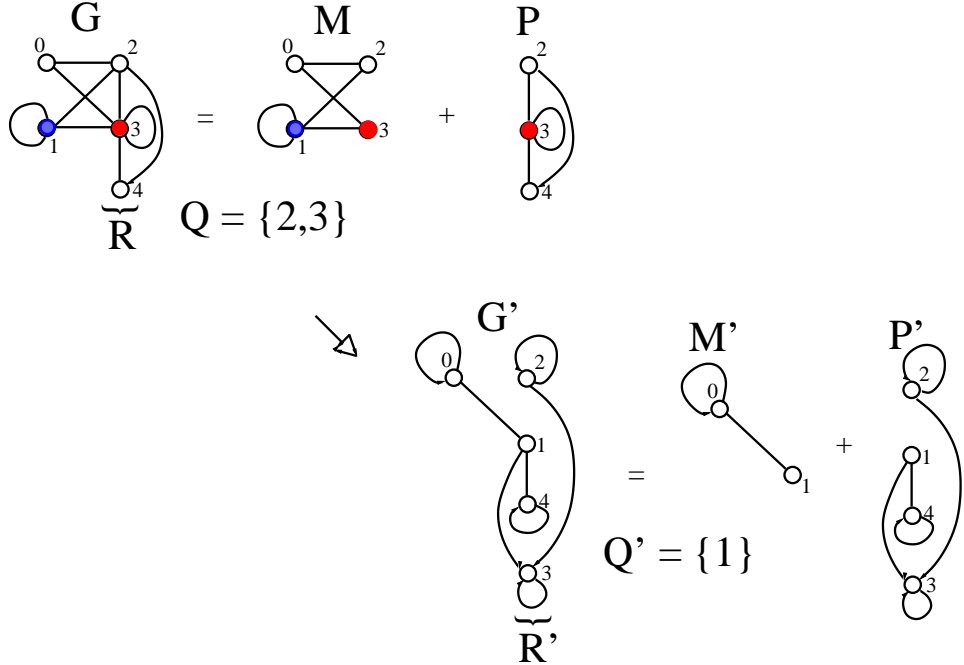}
\ce
\caption{The action of $H_3$ on a generalised two-graph state}
\label{figHvQ}
\end{figure}

We now describe the action of $N_v$ on a generalised two-graph state.

\begin{thm}
Let $\left| \psi \right >$ be a generalised two-graph state over $n$ qubits, represented by the graph-set-set
$(G,{\cal{R}},{\cal{Q}}) \equiv (M,P,{\cal{Q}})$.
Let $v \in \{0,1,\ldots,n-1\}$. 
Then $\left| \psi' \right > = N_v\left| \psi \right >$ is also a generalised two-graph state
and can be described by the graph-set-set $(G',{\cal{R}}',{\cal{Q}}') \equiv (M',P',{\cal{Q}}')$, where \\

\noindent $(G',{\cal{R}}',{\cal{Q}}') = N_v(G,{\cal{R}},{\cal{Q}})$:
$$
\begin{array}{ll}
\left \{
\begin{array}{l}
{\cal{R}}' = {\cal{R}} \cup \{v\}, \\
(G',{\cal{Q}}'') = (G,{\cal{Q}})^v, \\
{\cal{Q}}' = {\cal{Q}}'' \setminus \{v\},
\end{array} \right . & \mf \mbox{ if }v \in {\cal{L}}, \\ \\
\left \{
\begin{array}{l}
\m{if } {\cal{Q}}_v = 1 \\
\mf {\cal{R}}' = {\cal{R}} \setminus \{v\}, \\
\mf G' = G + \Delta_{\{v\}}, \\
\mf {\cal{Q}}' = {\cal{Q}} \setminus \{v\}, \\
\m{else } \\
\mf {\cal{R}}' = {\cal{R}}, \\
\mf (G'',{\cal{Q'}}) = (G,{\cal{Q}})^v \\
\mf G' = G'' + \Delta_{{\cal{B}}_v^G}
\enspace,
\end{array} \right . & \mf \mbox{ if }{\cal{B}}_v^G \subseteq {\cal{R}} \\ \\
\left \{
\begin{array}{l}
\m{assign } w  \in {\cal{N}}_v^M, \\
(G'',{\cal{R}}'',{\cal{Q}}'') = \m{swp}(G,{\cal{R}},{\cal{Q}},w,v), \\
{\cal{R}}' = {\cal{R}}'' \cup \{v\}, \\
(G',{\cal{Q}}''') = (G'',{\cal{Q}}'')^v, \\
{\cal{Q}}' = {\cal{Q}}''' \setminus \{v\},
\end{array} \right . & \mf \mbox{ if }v \in {\cal{R}}, {\cal{B}}_v^G \not\nsubseteq {\cal{R}}.
\end{array} $$
\label{NvAction}
\end{thm}
\begin{proof}
Section \ref{proofs}.
\end{proof}

We now describe the action of the inverse of $N_v$ on a generalised two-graph state. This is important for 
computational reasons, as it allows us
to compute spectral measures such as the $L_j$ norm and the
{\em{Clifford merit factor}} \cite{Par:CMF} of a 
graph state (section \ref{Spectral})
by using a Gray code ordering on successive actions of $H$ and $N$ on each qubit,
thereby avoiding the problem of having to
store all the graphs from every step.

\begin{lem} Let $\left| \psi \right >$ be a generalised two-graph state over $n$ qubits, represented by the graph-set-set
$(G,{\cal{R}},{\cal{Q}}) \equiv (M,P,{\cal{Q}})$.
Let $(G'',{\cal{R}}'',{\cal{Q}}'') \equiv (M'',P'',{\cal{Q}}'')$ be the
generalised two-graph state resulting from the
application of $H_v$ to $(G,{\cal{R}},{\cal{Q}})$.
Let $v \in \{0,1,\ldots,n-1\}$. Then the action of $N^{-1}$ on $v$ is the graph-set-set 
$(G',{\cal{R}}',{\cal{Q}}') \equiv (M',P',{\cal{Q}}')$, where:

$$
\begin{array}{ll}
\left\{\begin{array}{l}
G'=G''+C_{{\cal N}_v^{G''}}+\Delta_{{\cal N}_v^{G''}}+\Gamma_{vv}^{G''}\Delta_{{\cal N}_v^{G''}}+\Delta_{{\cal
Q}''\cap{\cal N}_v^{G''}}\\
{\cal Q}'={\cal Q}''\ominus{\cal N}_v^{G''} \end{array}\right. &  \m{ if } v\in{\cal L}''\\ \\
\left\{\begin{array}{l}
G'=G''+\Delta_{v}\\
{\cal Q}'={\cal Q}''\cup\{v\} \end{array}\right. & \m{ if } v\in{\cal R}'' \m{ and } {\cal Q}_v''=0\\ \\
\left\{\begin{array}{l}
G'=G''\\
{\cal Q}'={\cal Q}''\setminus\{v\} \end{array}\right. & \m{ if } v\in{\cal R}'' \m{ and } {\cal Q}_v''=1 \\
{\cal{R}}' = {\cal{R}}''. &
\end{array}$$
 \label{invN}\end{lem}
\begin{proof}
Section \ref{proofs}.
\end{proof}

\begin{thm} Let $\left| \psi \right > =mi^p$ be a generalised two-graph state over $n$ qubits. Then
there always exists a graph state, $\left| \psi \right >'$, such that
$\left| \psi \right >' = U\left| \psi \right >$, where $U \in {\bf{C_n}}$.
\label{tgrapheq}
\end{thm}
\begin{proof}
A generalised two-graph state, $(G,{\cal{R}},{\cal{Q}})$ is always locally-equivalent to a two-graph state,
$(G,{\cal{R}}) = (G,{\cal{R}},{\cal{Q}}')$, via the action of some unitary in ${\bf{D_n}}$, where
${\cal{Q}}' = \emptyset$. The theorem then follows from theorem \ref{twographeq}.
\end{proof}

\begin{cor} (theorem \ref{tgrapheq})
The generalised two-graph state is a stabilizer state, and vice-versa.
\label{corthm}
\end{cor}
\begin{proof}
From \cite{Gras:QG,Sch:QG} and theorem \ref{tgrapheq},
all stabilizer states and all generalised two-graph states are graph states, via
the action of unitaries from the local Clifford group, and such action is reversible.
\end{proof}


\subsection{The actions of  $\lambda$ and $\lambda^2$}
\label{lambda}

We have described the action of $N$ and $H$ on the generalised two-graph state. It is trivial to convert these
actions to the actions of $\lambda$ and $\lambda^2$ on the state, respectively, remembering that, in this paper,
global multiplicative constants are ignored. Explicitly, \\

$$ \lambda_v(G,{\cal{R}},{\cal{Q}}) = N_v(G,{\cal{R}},{\cal{Q}}), \mf \mf
\lambda_v^2(G,{\cal{R}},{\cal{Q}}) = S_v^{-1}H_v(G,{\cal{R}},{\cal{Q}}) = N_v^2(G,{\cal{R}},{\cal{Q}}). $$

\section{Canonisation}
\label{canon}
For some (generalised) two-graph state, $mi^p$, as represented by $(G,{\cal{R}},{\cal{Q}})$ over $n$ qubits,
let ${\cal{L}} \ne \emptyset$. Then there is a set of equivalent representations for the same
state. For purposes of comparison,
it is desirable to find a canonical representative from each set of equivalent representations. In this section
we provide a simple algorithm to obtain, from an arbitrary (generalised) two-graph state, a canonical representative.

\begin{df}
A generalised two-graph state, $(G,{\cal{R}},{\cal{Q}})$, is defined to be {\em{canonised}} if
$v < u$, $\forall (v,u) \in ({\cal{L}},{\cal{N}}_v^G)$.
\label{canonised}
\end{df}

Observe that such a canonical form is unique and, given such a unique $M$ graph, the $P$ graph and $Q$ set are also
unambiguously fixed. We now describe the process of canonisation of a generalised two-graph state.

\begin{lem}
Let $\left| \psi \right >$ be a generalised two-graph state over $n$ qubits, as represented by
$(G,{\cal{R}},{\cal{Q}})$. Then we can obtain a canonical representation of 
 $\left| \psi \right >$, as represented by $(G_n,{\cal{R}}_n,{\cal{Q}}_n)$,
by following these steps, where $\min({\cal{A}})$ means the minimum integer in set ${\cal{A}}$: \\

\noindent $(G_n,{\cal{R}}_n,{\cal{Q}}_n) = $ {\bf{canon}}$(G,{\cal{R}},{\cal{Q}})$:

\vspace{3mm} 

$$ \begin{array}{l}
\left \{
\begin{array}{l}
\m{Set } (G_0,{\cal{R}}_0,{\cal{Q}}_0)=(G,{\cal{R}},{\cal{Q}}), \m{ and set } i = 0 \\
\m{while } \exists v \in {\cal L}_i \m{ such that } v > \min({\cal N}_v^{G_i}) \\
\mf \mf w = \min({\cal N}_v^{G_i}) \\
\mf \mf (G_{i+1},{\cal{R}}_{i+1},{\cal{Q}}_{i+1}) =  \m{swp}(G_i,{\cal{R}}_i,{\cal{Q}}_i,v,w) \\
\mf \mf i \leftarrow i + 1.
\end{array} \right .
\end{array} $$
Each call to {\bf canon} will have worst-case complexity $O(|{\cal L}|^2)$.
\label{norm}
\end{lem}
\begin{proof}
Section \ref{proofs}.
\end{proof}

One can apply canonisation to the generalised two-graph state after each application of $N$ or $H$, if required.

\section{Spectral Analysis of the Graph State}
\label{Spectral}
We now briefly demonstrate the usefulness of the two-graph representation by
computing various {\em{$L_j$-norms}} of the graph
state wrt the local Clifford group. Let $\left| \psi \right >$ be a generalised two-graph state over $n$ qubits.
For $U \in {\bf{C_n}}$, let $\left| \psi_U \right > = U\left| \psi \right >$. The $L_j$-norm of $\left| \psi \right >$
is given by
$$ || \left| \psi \right > ||_j = \left ( 2^{-n}\sum_{{\bf{x}} \in {\mathbb{Z}}_2^n}
(2^{\frac{n}{2}}|\left| \psi \right >_{\bf{x}}|)^j \right )^{\frac{1}{j}}
= 2^{n(\frac{1}{2} - \frac{1}{j})} \left ( \sum_{{\bf{x}} \in {\mathbb{Z}}_2^n} |\left| \psi \right >_{\bf{x}}|^j \right )^{\frac{1}{j}}. $$

We wish to compute the $L_j$-norm over every state generated by the action of the local Clifford group on
$\left| \psi \right >$. However, as these norms only depend on a summary of
powers of magnitudes, it suffices to compute
the $L_j$-norm over every state generated by the action of $\{I,H,N\}^{\otimes n}$ on $\left| \psi \right >$, as
the action of matrices from ${\bf D_n}$ on the state does
not affect coefficient magnitudes; that is, let
$U=U_DU_T$, with $U_D\in {\bf D_n}$ and $U_T\in \{I,H,N\}^{\otimes n}$: then $|| \left| \psi_U \right > ||=
|| \left| \psi_{U_DU_T} \right > ||=|| \left| \psi_{U_T} \right > ||$. Thus
$$ || \left| \psi \right > ||_{{\bf{C_n}},j} =\left(\frac{24^{-n}}{8}\sum_{U\in {\bf{C_n}}}|| \left| \psi_U \right >
||_j^j\right)^{\frac{1}{j}}=
\left ( 3^{-n} \sum_{U \in \{I,H,N\}^{\otimes n}} || \left| \psi_U \right > ||_j^j \right )^{\frac{1}{j}}
 = 2^{\frac{n}{2}} \left (6^{-n}
\sum_{\begin{tiny} \begin{array}{c} {\bf{x}} \in {\mathbb{Z}}_2^n \\ U \in \{I,H,N\}^{\otimes n} \end{array}\end{tiny}}
|\left| \psi_U \right >_{\bf{x}}|^j \right )^{\frac{1}{j}}. $$

Normalisation of the pure state ensures that
$ || \left| \psi \right > ||_2 = || \left| \psi \right > ||_{{\bf{C_n}},2} = 1$, by Parseval's theorem.

Let $\left| \psi_U \right >$ be represented by the graph-set-set
$(G_U,{\cal{R}}_U,{\cal{Q}}_U)$, where
${\cal{L}}_U = {\cal{V}} \setminus {\cal{R}}_U$. Then one can show that,
$$ || \left| \psi_U \right > ||_j = 2^{\frac{(j-2)|{\cal{L}}_U|}{2j}}. $$
Therefore,
$$ || \left| \psi \right > ||_{{\bf{C_n}},j} =
\left ( 3^{-n} \sum_{U \in \{I,H,N\}^{\otimes n}} 2^{\frac{(j-2)|{\cal{L}}_U|}{2}} \right )^{\frac{1}{j}}. $$
In other words, $|| \left| \psi \right > ||_{{\bf{C_n}},j}$ can be efficiently computed by keeping track of the size of
${\cal{L}}_U$ after each successive action of $H$ and $N$ on the two-graph state. In particular, although the evaluation
is theoretically over all $24^n \times 8$ transforms represented by the local Clifford group, we obtain the same evaluation by only
considering the $3^n$ transforms represented by $\{I,H,N\}^{\otimes n}$, which is an exponential improvement in computational complexity.

Using the {\em{Database of Self-Dual Quantum Codes}} \cite{DanDat}
we classify all inequivalent graph states according to
their $L_j$ norms wrt ${\bf{C_n}}$, up to $n = 7$ qubits, as
shown in table 1
for $j = 3$ and $j = 4$, where the norm is
$|| \left| \psi \right > ||_{{\bf{C_n}},j}$.
One can expect the entanglement of the graph state to be higher if
$|| \left| \psi \right > ||_{{\bf{C_n}},j}$ is lower.
In \cite{Par:CMF}, the so-called {\em{Clifford merit factor}} (${\cal{CMF}}$) was
proposed as a suitable measure of entanglement for a graph state, where
$$ {\cal{CMF}}(\left| \psi \right >) = \frac{1}{|| \left| \psi \right > ||_{{\bf{C_n}},4}^4 - 1}. $$
One can expect the entanglement of a graph state to be higher if the
CMF of a graph state is higher.
Moreover, it was proved in \cite{Par:CMF} that the expected value of
$\frac{1}{\cal{CMF}}$ for a random graph state,
as $n \rightarrow \infty$, is $1$
\begin{footnote}{
Assumes all graphs are equally likely.}
\end{footnote}. This is suggested as, at least, reasonable
by the results of table 1
as $|| \left| \psi \right > ||_{{\bf{C_n}},4}^4$
for a random graph state
could well approach $2$ from below as $n \rightarrow \infty$.

\begin{table}[htb]
\begin{tiny}
\begin{center}
$ \begin{array}{cc}
\begin{array}{|c|c|c|}
\hline
n & || \left| \psi \right > ||_{{\bf{C_n}},3} & \m{frequency} \\ \hline
\hline
1 & 1.000000 & 1 \\
\hline
\m{average} & 1.000000 & 1 \\ \hline
\hline
2 & 1.000000 & 1 \\
\hline
\m{average}  & 1.000000 & 1 \\ \hline
\hline
3 & 1.079871 & 1 \\
\hline
\m{average}  & 1.079871 & 1 \\ \hline
\hline
4 & 1.035744 & 1 \\
\cline{2-3}
& 1.020167 & 1 \\
\hline
\m{average}  & 1.027955 & 2 \\ \hline
\hline
5 & 1.067977 & 1 \\
\cline{2-3}
& 1.040834 & 1 \\
\cline{2-3}
& 1.030604 & 1 \\
\cline{2-3}
& 1.020167 & 1 \\
\hline
\m{average}  & 1.039895 & 4 \\ \hline
\hline
6 & 1.106649 & 1 \\
\cline{2-3}
& 1.071174 & 1 \\
\cline{2-3}
& 1.059898 & 1 \\
\cline{2-3}
& 1.053345 & 1 \\
\cline{2-3}
& 1.047544 & 1 \\
\cline{2-3}
& 1.046710 & 1 \\
\cline{2-3}
& 1.040834 & 2 \\
\cline{2-3}
& 1.034036 & 1 \\
\cline{2-3}
& 1.027148 & 1 \\
\cline{2-3}
& 1.020167 & 1 \\
\hline
\m{average}  & 1.049849 & 11 \\ \hline
\hline
7 & 1.150213 & 1 \\
\cline{2-3}
& 1.108636 & 1 \\
\cline{2-3}
& 1.089457 & 1 \\
\cline{2-3}
& 1.085332 & 1 \\
\cline{2-3}
& 1.078038 & 1 \\
\cline{2-3}
& 1.075408 & 1 \\
\cline{2-3}
& 1.073824 & 1 \\
\cline{2-3}
& 1.067977 & 1 \\
\cline{2-3}
& 1.063683 & 2 \\
\cline{2-3}
& 1.059898 & 1 \\
\cline{2-3}
& 1.059355 & 1 \\
\cline{2-3}
& 1.056085 & 1 \\
\cline{2-3}
& 1.055538 & 1 \\
\cline{2-3}
& 1.051694 & 2 \\
\cline{2-3}
& 1.051143 & 1 \\
\cline{2-3}
& 1.047266 & 1 \\
\cline{2-3}
& 1.043080 & 1 \\
\cline{2-3}
& 1.042800 & 1 \\
\cline{2-3}
& 1.038578 & 1 \\
\cline{2-3}
& 1.034036 & 3 \\
\cline{2-3}
& 1.033751 & 1 \\
\cline{2-3}
& 1.029455 & 1 \\
\hline
\m{average}  & 1.060719 & 26 \\ \hline
\hline
\end{array} &
\begin{array}{|c|c|c|c|}
\hline
n & || \left| \psi \right > ||_{{\bf{C_n}},4} & \m{CMF} & \m{frequency} \\ \hline
\hline
1 & 1.074570 & 3.000000 & 1 \\
\hline
\m{average} 1 & 1.074570 & & 1 \\ \hline
\hline
2 & 1.074570 & 3.000000 & 1 \\
\hline
\m{average}  & 1.074570 & & 1 \\ \hline
\hline
3 & 1.240806 & 0.729730 & 1 \\
\hline
\m{average}  & 1.240806 & & 1 \\ \hline
\hline
4 & 1.154701 & 1.285714 & 1 \\
\cline{2-4}
& 1.121195 & 1.723404 & 1 \\
\hline
\m{average}  & 1.137948 & & 2 \\ \hline
\hline
5 & 1.223202 & 0.807309 & 1 \\
\cline{2-4}
& 1.165247 & 1.185366 & 1 \\
\cline{2-4}
& 1.143857 & 1.404624 & 1 \\
\cline{2-4}
& 1.121195 & 1.723404 & 1 \\
\hline
\m{average}  & 1.163375 & & 4 \\ \hline
\hline
6 & 1.304643 & 0.527115 & 1 \\
\cline{2-4}
& 1.229154 & 0.779679 & 1 \\
\cline{2-4}
& 1.204803 & 0.903346 & 1 \\
\cline{2-4}
& 1.192052 & 0.981157 & 1 \\
\cline{2-4}
& 1.178878 & 1.073638 & 2 \\
\cline{2-4}
& 1.165247 & 1.185366 & 2 \\
\cline{2-4}
& 1.151120 & 1.323049 & 1 \\
\cline{2-4}
& 1.136453 & 1.496920 & 1 \\
\cline{2-4}
& 1.121195 & 1.723404 & 1 \\
\hline
\m{average}  & 1.184334 & & 11 \\ \hline
\hline
7 & 1.396589 & 0.356595 & 1 \\
\cline{2-4}
& 1.307925 & 0.519108 & 1 \\
\cline{2-4}
& 1.266787 & 0.634833 & 1 \\
\cline{2-4}
& 1.259527 & 0.659331 & 1 \\
\cline{2-4}
& 1.244619 & 0.714472 & 1 \\
\cline{2-4}
& 1.236959 & 0.745653 & 2 \\
\cline{2-4}
& 1.221198 & 0.816959 & 1 \\
\cline{2-4}
& 1.213084 & 0.857984 & 2 \\
\cline{2-4}
& 1.204803 & 0.903346 & 2 \\
\cline{2-4}
& 1.196347 & 0.953772 & 2 \\
\cline{2-4}
& 1.187709 & 1.010162 & 3 \\
\cline{2-4}
& 1.178878 & 1.073638 & 1 \\
\cline{2-4}
& 1.169844 & 1.145626 & 2 \\
\cline{2-4}
& 1.160595 & 1.227962 & 1 \\
\cline{2-4}
& 1.151120 & 1.323049 & 4 \\
\cline{2-4}
& 1.141405 & 1.434098 & 1 \\
\hline
\m{average}  & 1.207200 & & 26 \\ \hline
\hline
\end{array}
\end{array}$
\label{NormsTab}
\end{center}
\caption{$|| \left| \psi \right > ||_{{\bf{C_n}},3}$ and
$|| \left| \psi \right > ||_{{\bf{C_n}},4}$ norms for graph states of $n = 1$ to $7$ vertices}
\end{tiny}
\end{table}

We can also compute the $L_{\infty}$ norm of a graph state wrt the local Clifford group, where,
$$ || \left| \psi \right > ||_{{\bf{C_n}},\infty} = 2^{\left ( \sup_{U \in \{I,H,N\}^{\otimes n}}(|{\cal{L}}_U|) \right )/2}, $$
and (potentially) ranges from $1$ to $2^{n/2}$ (although, for connected graphs, neither the `ideal' lower bound
or the worst-case upper-bound are ever reached). In
\cite{DanQECC} the PAR$_{IHN}$, of a graph state is computed, where
PAR$_{IHN}(\left| \psi \right >) = || \left| \psi \right > ||_{{\bf{C_n}},\infty}^2$,
and where $n - \log_2(\m{PAR}_{IHN})$ gives a lower bound on
the entanglement of the graph state as measured by the log form of the geometric measure
\cite{Geom}, which is an {\em entanglement monotone} \cite{Vidal:Mon}.
This lower bound is shown to be tight for a graph state with a bipartite graph in its
LC orbit \cite{Par:QE,DanQECC}.
The method used in
\cite{DanQECC} to compute PAR$_{IHN}$ looked for the independent set of largest size over the set of graphs in the
LC orbit of $\left| \psi \right >$. It is evident that $\sup_{U \in \{I,H,N\}^{\otimes n}}(|{\cal{L}}_U|)$ is equal
to the size of this largest independent set. So we do not strictly need the two-graph form to compute the
$L_{\infty}$-norm of the graph state, but can make do with LC over the graph state. However,
we then require to search for the largest independent set in each graph in the LC orbit.
In contrast, if we use the two-graph representation to compute the $L_{\infty}$-norm of the graph state then we
identify an independent set in the current graph wrt $U$ as being the set ${\cal{L}}_U$. Thus the two-graph
representation implicitly encodes and keeps track of the independent sets in the graphs in
the LC orbit
of the graph state. The search techniques of \cite{DanQECC} and this
paper are of approximately equal computational complexity. Results for PAR$_{IHN}$ for graph states
are provided in \cite{DanQECC}.
In figure \ref{PARIHNTab} we plot the expected PAR$_{IHN}$ of a graph state of varying
density, where the `density' indicates the percentage probability that a given
edge exists. From figure \ref{PARIHNTab} we conclude that very dense and very sparse graphs represent
graph states with relatively high values of PAR$_{IHN}$, which translates to a relatively low
lower bound on the geometric measure of entanglement. Therefore, as one might expect, it
appears that graph states of density around $0.5$ should maximise the lower bound on the geometric measure of entanglement.

\begin{figure}[!h]
\cb
\includegraphics[width=5in]{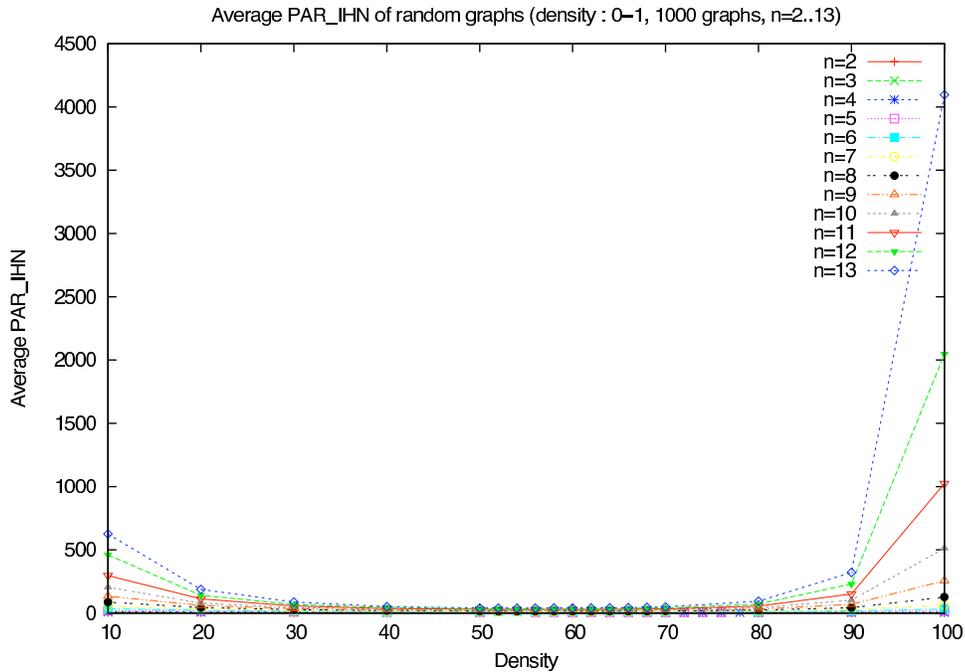}
\ce
\caption{Expected PAR$_{IHN}$ of random graph states of $n = 2$ to $13$ vertices
versus graph density, $10 - 100$\%}
\label{PARIHNTab}
\end{figure}

In figure \ref{manynorms}
we compute the expected value for an $L_j$ norm, $2 \le j < 16$, for a random graph
state of density 50\%. The horizontal lines are the results for the $L_{\infty}$ norm,
where $|| \left| \psi \right > ||_{{\bf{C_n}},\infty} = \sqrt{\m{PAR}_{IHN}}$. The results indicate that the
$L_{\infty}$ norm is approached from below by the $L_j$ norm as $j \rightarrow \infty$ (it is not so
difficult to prove this). The results also indicate that the relationship between
expected
$|| \left| \psi \right > ||_{{\bf{C_n}},\infty}$ and $n$ is marginally superlinear, at least for small numbers
of vertices.

\begin{figure}[!h]
\cb
\includegraphics[width=5in]{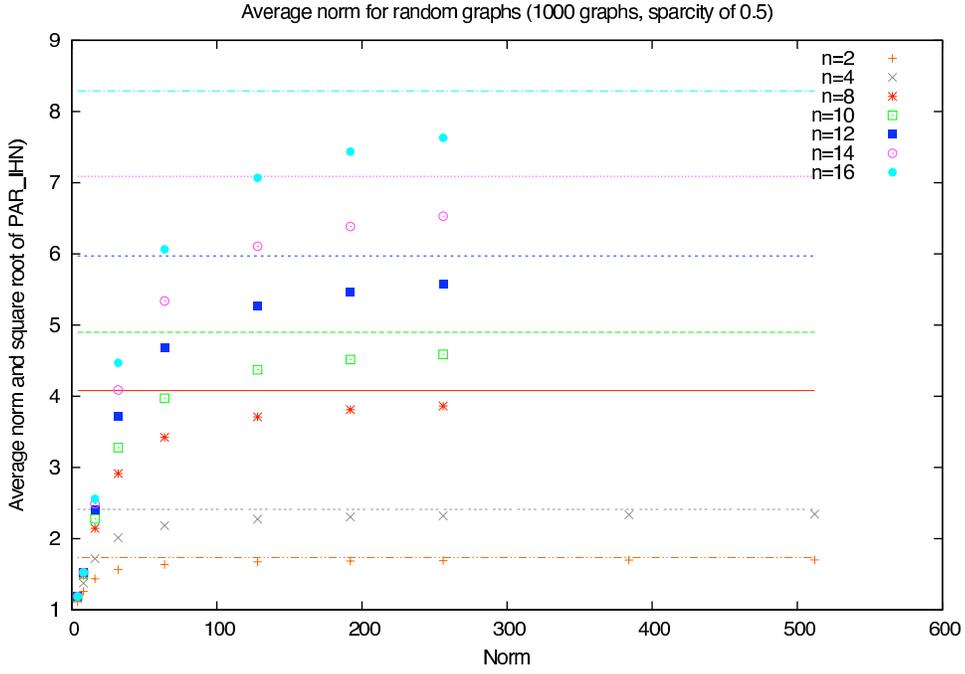}
\ce
\caption{Expected $L_j$ norm and $\sqrt{\m{PAR}_{IHN}}$ versus $j$ for
$n$-vertex random graph states}
\label{manynorms}
\end{figure}

\section{Appendix: Proofs}\label{proofs}


For an integer set, ${\cal S}$,
denote $x_{\cal S} = \sum_{j \in {\cal S}} x_j$, and let $x_a \in f$, $x_a \not \in f$, mean
that $f$ is or is not dependent on $x_a$, respectively.
We denote $p_a = p_{x_v = a}$, $m_a = m|_{x_v = a}$, for $a \in {\mathbb{Z}}_2$.

\subsection{Proofs for section \ref{TGS}}

\begin{proof} (of lemma \ref{SwpAction})
We write $m=r(x_v+x_w+h_v)\
\prod_{t \in {\cal N}_w^M \setminus \{v\}} (x_t+x_w+h_t)$, with $x_v,x_w\notin r$,
for $h_k= \Gamma_{M_{kk}}+1+ \sum_{j \in {\cal N}_k^M \setminus \{w\}} x_j$.
We also write
$p=x_w(x_{{\cal N}_w^P} + \Gamma_{P_{ww}}) +p|_{x_w=0}$. We want to
interchange the roles of $v$ and $w$ by re-factoring $m$ and by substituting $x_w=x_v+h_v+1$ in the remaining terms
that involve $x_w$. Thus
$m' = r(x_w+x_v+h_v)\prod_{t \in {\cal N}_w^M \setminus \{v\}}(x_t+x_v+h_v+h_t+1)$,
and $p'=(x_v+h_v+1)(x_{{\cal N}_w^P} + \Gamma_{P_{ww}}) +p|_{x_w=0}$.
From the form of $m$, $p$, and $m'$, $p'$, where ${\cal R}' = {\cal R} \cup \{v\} \setminus \{w\}$, we
obtain the graph equations,
$$ \begin{array}{l}
M' = M - K_{\{v\},{\cal N}_v^M} - K_{{\cal N}_w^M \setminus \{v\},\{w\}} - \Gamma_{M_{vv}}\Delta_{\{v\}}
 + K_{\{w\},{\cal B}_v^M \setminus \{w\}}
 + K_{{\cal N}_w^M \setminus \{v\},{\cal B}_v^M \setminus \{w\}}
  + \Gamma_{M_{vv}} ( \Delta_{\{w\}} + \Delta_{{\cal N}_w^M \setminus \{v\}}), \\
P' = P - K_{\{w\},{\cal N}_w^P} - \Gamma_{P_{ww}}\Delta_{\{w\}} + K_{{\cal B}_v^M \setminus \{w\},{\cal N}_w^P}
 + \Gamma_{M_{vv}}\Delta_{{\cal N}_w^P} + \Gamma_{P_{ww}}\Delta_{{\cal B}_v^M \setminus \{w\}}.
\end{array} $$
Rearranging,
$$ \begin{array}{l}
M' = M + K_{\{v\},{\cal N}_v^M} + K_{{\cal N}_w^M \setminus \{v\},{\cal B}_v^M} + \Gamma_{M_{vv}}\Delta_{{\cal B}_w^M \setminus \{v\}}, \\
P' = P + K_{{\cal B}_v^M,{\cal N}_w^P} + K_{\{w\},{\cal B}_v^M \setminus \{w\}} + \Gamma_{M_{vv}}\Delta_{{\cal N}_w^P}
 + \Gamma_{P_{ww}}\Delta_{{\cal B}_v^M \setminus \{w\}}.
\end{array} $$
Combining and simplifying,
$$ \begin{array}{ll}
G' = M' + P' & = G + K_{{\cal N}_w^M,{\cal B}_v^M} + \Delta_{\{v\}} + K_{{\cal B}_v^M,{\cal B}_w^P}
 + \Delta_{\{w\}} + \Gamma_{G_{vv}}\Delta_{{\cal B}_w^G} + \Gamma_{G_{ww}}\Delta_{{\cal B}_v^G} \\
   & = G + K_{{\cal B}_v^G,{\cal B}_w^G} + \Delta_{\{v,w\}} + \Gamma_{G_{vv}}\Delta_{{\cal B}_w^G} + \Gamma_{G_{ww}}\Delta_{{\cal B}_v^G} = G^{vw}.
\end{array} $$
\end{proof}

To prove theorem \ref{TGSAction} we require the following lemma.

\begin{lem} Let $m$ and $p$ be Boolean functions. Then, 
\beg  \sqrt{2}H_v m(-1)^p= m_0(-1)^{p_0}+m_1(-1)^{p_1+x_v}.
\label{hj}\eeg
\label{hvmp}\end{lem}

\begin{proof} (lemma \ref{hvmp})
Without loss of generality we set $v=n-1$. Then we write the
$2^n \times 1$ vector 
 $$m(-1)^p=\left( \begin{tiny} \begin{array}{c}
 m_0(-1)^{p_0}\\
 m_1(-1)^{p_1} \end{array} \end{tiny} \right). $$
Using the equality $\sqrt{2}H = X + Z$,
$$ H_nm(-1)^p = \frac{1}{\sqrt{2}} \left ( X_{n-1}
\left( \begin{tiny} \begin{array}{c}
 m_0(-1)^{p_0}\\
 m_1(-1)^{p_1} \end{array} \end{tiny} \right) +
 Z_{n-1}
 \left( \begin{tiny} \begin{array}{c}
 m_0(-1)^{p_0}\\
 m_1(-1)^{p_1} \end{array} \end{tiny} \right) \right )
  =
\frac{1}{\sqrt{2}} \left (
\left( \begin{tiny} \begin{array}{c}
 m_1(-1)^{p_1}\\
 m_0(-1)^{p_0} \end{array} \end{tiny} \right) +
 \left( \begin{tiny} \begin{array}{c}
 m_0(-1)^{p_0}\\
 -m_1(-1)^{p_1} \end{array} \end{tiny} \right) \right ). $$
 
\end{proof}

\begin{proof} (theorem \ref{TGSAction})

For
$x_v\notin m$ (i.e. ${\cal B}_v^G\subseteq{\cal R}$), we only need to show that
$\left|\psi'\right>=\frac{1}{\sqrt{2}}m'(-1)^{p'}=
H_vm(-1)^p$, where
\beg  m'=m(x_{{\cal B}_v^P}+1),\ \ \ p'=p_0 \enspace, \label{vnotinm} \eeg
as this implies that $M' = M + P_v$, $P' = P - P_v$,
$G' = M' + P' = G$, and $v \in {\cal L}'$, as required.
By lemma \ref{hvmp}, and given that $x_v\notin m$,
$\sqrt{2}H_v m(-1)^p=m(-1)^{p_0}(1+(-1)^{x_{{\cal B}_v^P}})=
\left\{\begin{array}{ccl}
0&\m{ if }& x_{{\cal B}_v^P}=1 \m{ mod }2\\
(-1)^{p_0}&\m{ if }& x_{{\cal B}_v^P}=0 \m{ mod }2,
\end{array}\right. $
thereby proving equation (\ref{vnotinm}) and the case where $x_v\notin m$.

\vspace{3mm}

For $v \in {\cal L}$, then $p_0 = p_1 = p$, and
$$ m_0 + (-1)^{x_v}m_1 = (-1)^{x_v(\Gamma_{M_{vv}} + x_{{\cal N}_v^G})}
\frac{m}{(\Gamma_{M_{vv}} + 1 + x_{{\cal B}_v^G})}. $$
Then, from lemma \ref{hvmp},
$$ m' = \frac{m}{(\Gamma_{M_{vv}} + 1 + x_{{\cal B}_v^G})}, \mf \m{ and }
\mf p' = p + x_v(\Gamma_{M_{vv}} + x_{{\cal N}_v^G}). $$
Therefore $M' = M - M_v$, $P' = P + M_v$ and, therefore,
$G' = M' + P' = G$, where $v \in {\cal R}'$, thereby proving the case
where $v \in {\cal L}$.

\vspace{3mm}

For $v\in{\cal R},\,{\cal B}_v\nsubseteq{\cal R}$, then, for
$\omega \in {\cal N}_v^M$, we first apply `swp' to interchange $v$ and $w$ so
that $v \in {\cal L}''$, where
$ m''(-1)^{p''} = m(-1)^p$. The case where $v\in{\cal R},\,{\cal B}_v\nsubseteq{\cal R}$
is then proved by showing that subsequently applying $H_v$ to
$(G'',{\cal R}'')$, where $v \in {\cal L}''$, obtains the result in the theorem,
and such a case has been proved above.

\end{proof}

\subsection{Proofs for section \ref{TGSgen}}

In the sequel we mix arithmetic, mod 2, and mod 4 so, to clarify
the formulas for equations that mix moduli, anything in square brackets is
computed mod 2. The $\{0,1\}$ result is then embedded in mod 4 arithmetic
for subsequent operations outside the square brackets.

We use the following lemma:
\begin{lem}
$$\sum_{i=1}^n [A_i] \ (\mbox{mod } 4)  =  [\sum_{i=1}^n A_i] + 2[\sum_{i < j} A_i A_j] \ (\mbox{mod } 4),\
\wh A_i \in {\mathbb{Z}}_2 \enspace.$$
\label{lem2} \end{lem}

\begin{proof} (of lemma \ref{SwpActiongen})
This lemma generalises lemma \ref{SwpAction}. Using the same notation as in the proof of lemma \ref{SwpAction},
we want to interchange the roles of $v$ and $w$ and,
as we define ${\cal R'}={\cal R}\cup\{v\}\setminus\{w\}$, we substitute $x_w=x_v+h_v+1$ where appropriate. The function $m'$
is the same as in the proof of lemma \ref{SwpAction}. For $p'$ we write
$$ p'=[x_v+h_v+1](2x_{{\cal N}_w^P}+ 2\Gamma_{P_{ww}} + {\cal Q}_w) + p|_{x_w=0}, $$
which is the same as in the proof of lemma \ref{SwpAction} apart from the term ${\cal Q}_w[x_v + h_v + 1]$.
The case were ${\cal Q}_w = 0$ is proven in lemma \ref{SwpAction}. For ${\cal Q}_w = 1$ we observe, from lemma \ref{lem2},
that
$$
[x_v + h_v + 1] = [x_{{\cal B}_v^M \setminus \{w\}} + \Gamma_{M_{vv}} ]
 = x_{{\cal B}_v^M \setminus \{w\}} + \Gamma_{M_{vv}} + 2 \sum_{i,j \in {\cal B}_v^M \setminus \{w\}, i < j} x_ix_j
  + 2\Gamma_{M_{vv}}x_{{\cal B}_v^M \setminus \{w\}}.
$$
The last equation can be interpreted graphwise as adding to the graph $G^{vw}$ the terms
$$ C_{{\cal B}_v^G \setminus \{w\}} + \Gamma_{G_{vv}}\Delta_{{\cal B}_v^G \setminus \{w\}} + \Delta_{{\cal Q} \cap {\cal B}_v^G \setminus \{w\}}, $$
and setting ${\cal Q}' = {\cal Q} \ominus {\cal B}_v^G$.
By definition \ref{ELC} we obtain ${\cal B}_v^G \setminus \{w\} = {\cal N}_w^{G^{vw}}$, and
$\Gamma_{G_{vv}} = \Gamma_{G_{ww}^{vw}}$. Substituting above we obtain
$G' = G^{vw} + C_{{\cal{N}}_w^{G^{vw}}} + \Gamma_{G^{vw}_{ww}} \Delta_{{\cal{N}}_w^{G^{vw}}} + \Delta_{{\cal{Q}} \cap
{\cal{N}}_w^{G^{vw}}}$, with ${\cal Q}'={\cal Q}\ominus B_w^{G^{vw}}$.
\end{proof}

In order to prove theorem \ref{HvAction}, we first state some spectral results.

\begin{lem}  Let $m$ be a Boolean function, and let $p:{\mathbb Z}_2^n\rightarrow{\mathbb Z}_4$. Then, 
\beg  \sqrt{2}H_v [m]i^p=[m_0]i^{p_0}+[m_1]i^{p_1+2x_v}.
\label{hjgen}\eeg
\end{lem}

\begin{proof} A trivial generalisation of the proof for lemma \ref{hvmp}.
\end{proof}

\begin{lem} Let $v\in{\cal L}$ and
let $s=\Gamma_{M_{vv}}+1+ x_{{\cal B}_v^M}$.
Then equation (\ref{hjgen}) can be rewritten as:
\beg \sqrt{2}H_v [m]i^p=[\frac{m}{s}]i^{p+2\Gamma_{M_{vv}}x_v+
2x_vx_{{\cal N}_v^M}} \enspace.
\eeg
\label{hvgenleft}
\end{lem}

\begin{proof} As $v\in{\cal L}$, $p_0=p_1=p$; writing $m=rs$, where $x_v \not\in r$, and
substituting for $m$ in equation (\ref{hjgen}), we get
$$ \sqrt{2}H_v [m]i^p=[r]i^p\left([s_0]+[s_1]i^{2x_v}\right)=
[r]i^p\left([\Gamma_{M_{vv}}+1+ x_{{\cal N}_v^M}]
+[\Gamma_{M_{vv}}+ x_{{\cal N}_v^M}]\right).
$$
Either $[s_0]=1$ or $[s_1]=1$, so 
$$ \sqrt{2}H_v [m]i^p=[r]i^{p+[\Gamma_{M_{vv}}+ x_{{\cal N}_v^M}]2x_v}=
[\frac{m}{s}]i^{p+2\Gamma_{M_{vv}}x_v+2x_vx_{{\cal N}_v^M}}. $$
\end{proof}

\begin{lem} Let $[m]i^p$ be a generalised two-graph state.
Let ${\cal B}_v^G\subseteq{\cal R}$, and let ${\cal
Q}_v=0$. Then equation (\ref{hjgen}) can be rewritten as:
\beg \sqrt{2}H_v [m]i^p=2[m][\Gamma_{P_{vv}}+1+ x_{{\cal B}_v^P}]i^{p_0} \enspace.
\eeg \label{hvgenright}
\end{lem}

\begin{proof} As ${\cal B}_v^G\subseteq{\cal R}$, we have $m_0=m_1=m$. Therefore we can rewrite equation
(\ref{hjgen}) as:
$$ \sqrt{2}H_v [m]i^p=[m]i^{p_0}\left(1+i^q\right),
$$
where $q = 2x_{{\cal B}_v^P} +2\Gamma_{P_{vv}}$.
The expression $1+i^q=0$ iff $q=2$ (mod 4); furthermore $q = 0$ or $2$ (mod 4), so
otherwise $1+i^q=2$. Thus we obtain a new term in the magnitude, namely the factor 
$[\Gamma_{P_{vv}}+1+ x_{{\cal B}_v^P}]$.
\end{proof}

\begin{proof} (theorem \ref{HvAction})
From lemma \ref{hvgenleft} we see that, for $v \in {\cal L}$,
$M' = M - M_v$, $P' = P + M_v$, and $v \in {\cal R}'$, and it follows that
$G' = M' + P' = G$.
From lemma \ref{hvgenright} we see that, for ${\cal B}_v^G\subseteq{\cal R}$, when ${\cal Q}_v = 0$,
then $M' = M + P_v$, $P' = P - P_v$, and $v \in {\cal L}'$, and it follows that $G' = M' + P' = G$.
For the case where $v\in{\cal R}$ and ${\cal B}_v^G\nsubseteq{\cal R}$, we need only to make a
swap to obtain $v\in{\cal L}'$, and then apply lemma \ref{hvgenleft}. We prove the remaining
case indirectly in lemma \ref{nonboolnotinm}, where the relevance of lemma \ref{nonboolnotinm}
to theorem \ref{HvAction} is proven by lemma \ref{interchange}.
\end{proof}

In order to prove theorem \ref{NvAction}, we first state some spectral results.

\begin{lem} \cite{thesis} Let $m$ be a Boolean function, and let
$p:{\mathbb Z}_2^n\rightarrow{\mathbb Z}_4$. Then, 
\beg  \sqrt{2}N_v [m]i^p=[m_0]i^{p_0}+[m_1]i^{p_1+2x_v+1}.
\enspace,\label{njgen}\eeg
\end{lem}


\begin{lem} Let
$v\in{\cal L}$
 and let $m = r(x_v+u+\Gamma_{G_{vv}})$, $x_v \not\in r$.
Then equation (\ref{njgen}) can be rewritten as:
\beg \sqrt{2}N_v [m]i^p=
[r]i ^{p_0+[u+\Gamma_{G_{vv}}](2x_{{\cal B}_{v}^P}+1)}=
[r]i ^{p_0+(u+\Gamma_{G_{vv}}+2\Gamma_{G_{vv}}u+
2\sum_{t,t'\in u,t < t'}x_tx_t')(2x_{{\cal B}_{v}^P}+1)}\enspace.\eeg

\label{nonboolinmN}
\end{lem}

\noindent


\begin{proof} Let $x_v\in m$ (i.e. ${\cal B}_v\nsubseteq{\cal
R}$). By equation (\ref{njgen}), $\sqrt{2}N_j [m]i^p=[m_0]i^{p_0}+[m_1]i^{p_1+2x_j+1}$.
Let $v\in{\cal L}$, so that ${\cal N}_v^P=\emptyset$. 
Writing $s=x_v+u+\Gamma_{G_{vv}}$, we obtain 
\beg \sqrt{2}N_v [m]i^p=[r]([s_0]i^{p_0}+[s_1]i^{p_1+2x_v+1})=[r]i^{p_0}([u+\Gamma_{G_{vv}}]+
[1+u+\Gamma_{G_{vv}}]i^{2x_v+1}) \enspace.\eeg
When $[u+\Gamma_{G_{vv}}]=0$,
$\sqrt{2}N_v [m]i^p=[r]i ^{p_0}i^{2x_{{\cal B}_{v}^P}+1}$; when 
$[u+\Gamma_{G_{vv}}]=1$, $\sqrt{2}N_v [m]i^p=[r]i ^{p_0}$. This can be summed up as
$\sqrt{2}N_v [m]i^p=[r]i ^{p_0+[u+\Gamma_{G_{vv}}](2x_{{\cal B}_{v}^P}+1)}$,
and the expansion follows from
lemma \ref{lem2}.
\end{proof}



\begin{lem} Let ${\cal B}_{v}^G\subseteq{\cal R}$ 
and let ${\cal Q}_v=0$.
Then equation (\ref{njgen}) can be rewritten as:
\beg  N_v [m]i^p=\displaystyle\frac{1+i}{\sqrt{2}}[m]i^{p+2\sum_{k <
k'}[m_km_{k'}]+(3+2\Gamma_{G_{vv}})\sum_k[m_k]+(3+2\Gamma_{G_{vv}})[x_v]+3\Gamma_{G_{vv}}}\eeg
where $x_{{\cal N}_v^P}=2[\sum_k m_k]+\Gamma_{G_{vv}}$.
 \label{nonboolnotinm}
\end{lem}

\begin{proof} Let $x_v\notin m$; then $m_0=m_1=m$, and therefore
$\sqrt{2}N_v[m]i^p=[m](i^{p_0}+i^{p_1+2x_v+1})
=[m]i^{p_0}(1+i^{2{\cal B}_v^G+2\Gamma_{G_{vv}}+1})$.
When  ${\cal Q}_v=0$, the coefficients of $2x_{{\cal B}_v^G}+2\Gamma_{G_{vv}}+1$
are in $\{1,3\}$, and 
so there are no solutions to $1+i^{2x_{{\cal B}_v^G}+2\Gamma_{G_{vv}}+1}=0$, and
this term is equal to $1+i$ when 
$2x_{{\cal B}_v^G}+2\Gamma_{G_{vv}}=0$, equal to $1-i$ otherwise. If we
divide by $1+i$, we get $[m]i^{p_0}i^{0}$ when
$2x_{{\cal B}_v^G}+2\Gamma_{G_{vv}}=0$, $[m]i^{p_0}i^{3}$ otherwise.
Using lemma \ref{lem2}, we obtain
$\sqrt{2}N_v[m]i^p=(1+i)[m]i^{p_0}i^{3/2(2x_{{\cal B}_v^G}+2\Gamma_{G_{vv}})}$,
and the lemma follows by observing that $x_{{\cal N}_v^G}=\sum_k m_k$.
\end{proof}

\noindent
{\bf Remark:} Note that for the case ${\cal B}_v\nsubseteq{\cal R}$ but $v\in{\cal R}$, we can swap with some element in the
neighbourhood to obtain the desired formula.

\begin{lem} Let $v\in{\cal R}$, and let ${\cal Q}_v=1$. Then, the action of $N_v$ (resp. $H_v$) on the two graph-state 
corresponding to $G$ is equal to the action of $H_v$ (resp. $N_v$) on the two-graph state corresponding to the
graph $G$ with a possible loop in $G$ at $v$ and ${\cal Q}_v=0$; moreover, the loop will appear iff $\Gamma_{P_{vv}}=0$ (resp. 
$\Gamma_{P_{vv}}=1$). \label{interchange}
\end{lem}

\begin{proof} $ \left(\begin{tiny}\begin{array}{rr}
1&i\\
1&-i
\end{array} \end{tiny}\right)\left(\begin{tiny}\begin{array}{rr}
1&0\\
0&\pm i
\end{array} \end{tiny}\right)=\left(\begin{tiny}\begin{array}{rr}
1&\mp 1\\
1&\pm 1
\end{array} \end{tiny}\right)=\left(\begin{tiny}\begin{array}{rr}
1&1\\
1&-1
\end{array} \end{tiny}\right)\left(\begin{tiny}\begin{array}{rr}
1&0\\
0&\mp 1
\end{array} \end{tiny}\right)$. \newline
Similarly,
$\begin{tiny} \left(\begin{array}{cc}
1&1\\
1&-1
\end{array}\right)\left(\begin{array}{rr}
1&0\\
0&\pm i
\end{array}\right)=\left(\begin{array}{rr}
1&\pm i\\
1&\mp i
\end{array}\right)=\left(\begin{array}{rr}
1&i\\
1&-i
\end{array}\right)\left(\begin{array}{rr}
1&0\\
0&\pm 1
\end{array}\right) \end{tiny}$.
\end{proof}

\begin{cor} Theorem \ref{NvAction}.
\end{cor}

\begin{proof} By lemmas \ref{nonboolinmN},  \ref{nonboolnotinm},  lemma
\ref{hvgenright}, 
and the proof of lemma \ref{SwpAction}.
\end{proof}

\begin{proof} (lemma \ref{invN})
We first observe that
$S_v^3[m]i^p = [m]i^{p + 3x_v}$. Moreover, $N^{-1} = S^3H$.
Thus, applying $N_v^{-1}$ to $[m]i^p$ is the same as
first applying $H_v$, then $S_v^3$, to $[m]i^p$.
 
Let $v\in{\cal L}''$: then $m''=r(x_v+\Gamma_{M_{vv}''}+1+ q)$, where
$q = x_{{\cal N}_v^{M''}}$.
Then, $S_v^3[m'']i^{p''}=[m'']i^{p''+3x_v}=
[m'']i^{p''+3[\Gamma_{M_{vv}''}+ q]}$. On the other hand, in mod 4,
$[\Gamma_{M_{vv}''}+ q]=\Gamma_{M_{vv}''}+ q
+2\Gamma_{M_{vv}''}q +2\sum_{j,k \in {\cal N}_v^{M''}, j < k} x_jx_k$. 
Then, 
$$ S_v^3[m'']i^{p''}=[m'']i^{p''+3\Gamma_{M_{vv}''}+3q
+2\Gamma_{M_{vv}''}q+2\sum_{j,k \in {\cal N}_v^{M''}, j < k} x_jx_k}\enspace.
$$

Let $v\in{\cal R}'',\,{\cal Q}_v''=0$, then we obtain an extra
loop in $G$ at $v$ and ${\cal Q}'={\cal Q}\cup\{v\}$.

When $v\in{\cal R}'',\,{\cal Q}_v''=1$, then the term $x_v$ cancels with $3x_v$ and
makes ${\cal Q}'={\cal Q}''\setminus\{v\}$.
\end{proof}

\subsection{Proofs for section \ref{canon}}

\begin{proof} (lemma \ref{norm})
For any generalised two-graph state, $(G,{\cal R},Q)$ there is always one unique equivalent canonised form,
$(G_c,{\cal R}_c,Q_c)$, such that the indices in set ${\cal L}_c$ are as small as possible. We first state and prove the
following lemma.
\begin{lem} For any uncanonised generalised two-graph state, $(G,{\cal R},Q)$, there always exists
at least one $v \in {\cal L}$ such that $v \not \in {\cal L}_c$ and $v > \min({\cal N}_v^G)$.
\label{biggerexists}
\end{lem}
\begin{proof} (lemma \ref{biggerexists})
By definition an uncanonised generalised two-graph state, $(G,{\cal R},Q)$, must contain at least
one $v \in {\cal L}$ such that $v \not \in {\cal L}_c$. We call such a $v$ `uncanonical'. Assume
that there is precisely one uncanonical element, $v$, contained in ${\cal L}$. We shall now assume
that $v < \min({\cal N}_v^G)$ and show, by contradiction, that such an assumption is impossible.
If $v < \min({\cal N}_v^G)$, then
there exists a codeword in the dual code associated with $m$ (ignoring loops in $M$) of the form
$ 00 \ldots 01 xx\ldots x$ where the leftmost $1$ occurs in position $v$ (numbering
positions from $0$ on the left). But we have assumed that $v \not \in {\cal L}_c$ so there must also
exist $|{\cal L}_c|$ other codewords in the dual code associated with $m$, also of the form
$ 00 \ldots 01 xx\ldots x$, where the left-most $1$ now occurs in position $u$, $\forall u \in {\cal L}_c$.
Thus, in total, we have $|{\cal L}_c| + 1$ codewords from the dual code. They are clearly pairwise linearly
independent so generate a linear space of size $2^{|{\cal L}_c| + 1}$. But the dual code associated with $m$
is only of size $2^{|{\cal L}_c|}$. This is a contradiction. The same argument can be generalised to the case
where more than one uncanonical element is contained in ${\cal L}$ and to the case where $M$ contains loops.
\end{proof}

By lemma \ref{biggerexists} we can always perform at least one `$swp$' at edge $vw$ on an uncanonised
generalised two-graph state, $(G,{\cal R},Q)$, where $v \in {\cal L}$, $w \in {\cal R}$, and $v > w$, so as to
produce a new generalised graph state, $(G',{\cal R}',Q')$, where $w \in {\cal L}'$ and $v \in {\cal R}'$.
It is then straightforward to see that one must obtain the canonised form after, at worst-case,
$|{\cal L}_c| \choose 2$ `$swp$s'.
\end{proof}

\end{document}